\documentclass[12pt,a4,english]{article}
\usepackage[latin9]{inputenc}
\usepackage{geometry}
\usepackage{verbatim}
\usepackage{amsmath}
\usepackage{amsfonts}
\usepackage{amsthm}
\usepackage{amssymb}
\usepackage{mathtools}
\usepackage{graphicx}
\usepackage{longtable}
\usepackage{setspace}
\usepackage{comment}
\usepackage{upgreek}

\usepackage[round,sort,authoryear]{natbib}
\usepackage[colorlinks=true,linkcolor=blue, citecolor=blue, urlcolor = blue]{hyperref}

\usepackage{titlesec}
\usepackage{babel}
\usepackage[toc,page]{appendix}
\usepackage{booktabs,caption}

\usepackage[labelfont=bf,textfont=md]{caption}
\usepackage[labelsep=space]{caption}

\usepackage{lmodern}
\usepackage[T1]{fontenc}

\DeclarePairedDelimiter{\floor}{\lfloor}{\rfloor}

\usepackage{bibentry}
\nobibliography*

\numberwithin{equation}{section}

\titleformat*{\section}{\large \bfseries}
\titleformat*{\subsection}{\normalsize \bfseries}
\titleformat*{\subsubsection}{\small \bfseries}

\newif\ifshow 
\showfalse 

\ifshow
  \includecomment{wrap}
\else
  \excludecomment{wrap} 
\fi

\theoremstyle{definition}
\newtheorem{theorem}{Theorem}
\newtheorem{definition}{Definition}
\newtheorem{assumption}{Assumption}

\newtheorem{lemma}{Lemma}
\newtheorem{example}{Example}
\newtheorem{proposition}{Proposition}
\newtheorem{remark}{Remark}
\newtheorem{corollary}{Corollary}

\newcommand\norm[1]{\left\lVert#1\right\rVert}

\begin{document}
\pagenumbering{roman}

\title{\Large \textbf{Estimating Conditional Value-at-Risk with Nonstationary Quantile Predictive Regression Models}\thanks{\textbf{Article history:} First draft 24 April 2022. Second draft 29 August 2023. Previously titled: "\textit{A Doubly IVX Corrected Estimator for Quantile Predictive Regressions with Nonstationary and Generated Regressors}". } }

\author{\textbf{Christis Katsouris}\thanks{Dr. Christis Katsouris is currently a Postdoctoral Researcher at the Faculty of Social Sciences, University of Helsinki, P.O. Box 17, Arkadiankatu 7 Helsinki, FI-00014 Finland. E-mail: \textcolor{blue}{\texttt{christis.katsouris@helsinki.fi}} } \\ \textit{University of Southampton}\\ $\&$  \textit{University of Helsinki} \\ \\ Working paper}

\date{\today}

\maketitle

\begin{abstract}
\vspace*{-0.25 em}
This paper develops an asymptotic distribution theory for an endogenous instrumentation approach in quantile predictive regressions when both generated covariates and persistent predictors are used. The generated covariates are obtained from an auxiliary quantile predictive regression model and the statistical problem of interest is the robust estimation and inference of the parameters that correspond to the primary quantile predictive regression in which this generated covariate is added to the set of nonstationary regressors. We find that the proposed doubly IVX corrected estimator is robust to the abstract degree of persistence regardless of the presence of generated regressor obtained from the first stage procedure. The asymptotic properties of the two-stage IVX estimator such as mixed Gaussianity are established while the asymptotic covariance matrix is adjusted to account for the first-step estimation error. 
\\

\\

\textbf{JEL} Classification: C12, C22
\end{abstract}

\newpage 

\setcounter{page}{1}
\pagenumbering{arabic}

\section{Introduction}

Risk monitoring of financial time series as well as portfolio performance measures when obtaining optimal portfolios can lead to inaccurate results in the case of non-Gaussian (asymmetric) return distributions. Thus, in order to better measure downside risk in a non-Gaussian setting usually the risk measures of Value-at-Risk and Expected-Shortfall are used to quantify uncertainty (see, \cite{lin2023portfolio}). Moreover, estimating risk measures such as the Conditional-Value-at-Risk, ($\mathsf{CoVaR}$),  entails a generated regressor problem, which occurs when an explanatory variable is constructed using structural econometric methods. We propose an estimation method in nonstationary time series regression  models with generated covariates, which has not seen much attention in the literature so far and thus requires the development of our own asymptotic theory. 

Inferential procedures based on generated regressors are studied in linear and conditional quantile models as in \cite{chen2021quantile} (see, \cite{pagan1984econometric}, \cite{doran1983relative}, \cite{hoffman1987two},  \cite{oxley1993econometric}, \cite{dufour2001finite} and \cite{chen2023standard}), although in the case of risk measures such as $\mathsf{VaR}$ and $\mathsf{CoVaR}$ a conditional quantile specification is required. The conventional approach for estimating these risk measures is to employ stationary quantile regressions (as in \cite{AB2016covar} and \cite{hardle2016tenet}). Several studies in the literature consider the estimation of risk measures using stationary time series environments either with parametric (e.g., see \cite{he2020inference}) or semiparametric approaches (e.g., see, \cite{wang2016conditional}). In this study,  we extend the econometric specifications for the joint estimation of the pair $( \mathsf{VaR}, \mathsf{CoVaR} )$ as discussed in the study of \cite{katsouris2023statistical} (see also \cite{katsouris2021optimal, katsouris2023quantile}), to the case of nonstationary quantile predictive regression models (see, \cite{lee2016predictive}), which implies using a local-to-unity parametrization for modelling the unknown form of persistence.
    
Specifically,  \cite{katsouris2023statistical} considers the statistical estimation for covariance-type matrices  with tail forecasts using node-specific quantile predictive regressions, extending the framework of \cite{katsouris2021optimal} who discusses identification and estimation issues of conditional quantile risk measures. However, both of these studies operate under the assumption of time series stationarity, thereby relying on OLS-based estimators (see, \cite{patton2019dynamic}). Thus, a suitable estimation approach under 
the presence of both a generated covariates and persistent regressors, for the joint estimation of the risk measure pair of $( \mathsf{VaR}, \mathsf{CoVaR} )$, remains an open problem. Several studies in the literature develop estimation and inference methods in nonstationary predictive regressionss robust to the unknown form of persistence, with a notable approach being the IVX filtration of \cite{phillipsmagdal2009econometric}. Motivated by the aforementioned issues, our research objective is to study the large-sample theory of a doubly IVX corrected estimator in nonstationary quantile predictive regression for the purpose of estimating risk measures. To the best of our knowledge, our approach that incorporates the time series properties of regressors (nonstationarity) when estimating risk measures in a low dimensional setting (small number of regressors relative to the sample size), is a novel contribution to the literature.

\newpage

As a result, when developing robust estimation and inference methods in econometric environments with generated regressors, it is crucial to obtain analytical and computational feasible expressions for the asymptotic variance of these potentially semiparametric estimators. In particular, \cite{rilstone1996nonparametric} investigates the aspect of generated covariate in a nonparameteric stationary time series environment.  Moreover, \cite{hahn2013asymptotic} study the asymptotic distribution of three-step estimators of a finite-dimensional parameter vector where the second step consists of one or more nonparametric regressions on a regressor that is estimated in the first step. Specifically, the authors using the  path-derivative method of \cite{newey1994asymptotic}, derive the contribution of the first-step estimator to the influence function is used to investigate the dual role that the first-step estimator plays in the second-step nonparametric regression. Moreover, \cite{bhattacharya2020quantile} consider the aspects of estimation and inference in stationary quantile regression models with generated covariates. Recently, \cite{ghosh2022doubly} develop a framework for semiparametric estimation where a doubly robust estimating function for a low-dimensional parameter is available, depending on two working models. Therefore, our framework proposes a computationally tractable two-step algorithm and an asymptotic theory analysis for the doubly IVX corrected estimator (which is found to robustify inference to the unknown form of persistence\footnote{The unknown form of persistence in time series data is modelled via the the local-to-unity parametrization introduced by the seminal study of \cite{phillips1987time, phillips1987towards} who established a unified framework for unit root and nearly integrated processes in time series regression models.}), when estimating risk measures under the assumption of statistical elicitability (see, \cite{fissler2016higher}). 

Our practical two stage estimation procedure permits to obtain estimates for the forecasts of the pair $( \mathsf{VaR}, \mathsf{CoVaR} )$ based on the econometric specifications of \cite{AB2016covar} (see also \cite{hardle2016tenet}), but under the assumption of time series nonstationarity using the local-to-unity parametrization to capture nearly nonstationary regressors as in \cite{lee2016predictive} and \cite{fan2019predictive}. During the first-stage procedure using the nonstationary quantile predictive regression model (QPR) and the local-to-unity parametrized regressors, we obtain quantile dependent model estimates that correspond to the one-period ahead forecasted $\mathsf{VaR}$ risk measure. During the second stage estimation procedure, the generated regressor that corresponds to the $\mathsf{VaR}$ is employed as an additional covariate to a different QPR model which is employed to obtain the one-period ahead forecasted $\mathsf{CoVaR}$ risk measure. Furthermore, because of the well-known \textit{endogeneity problem} (see, \cite{cai2014testing}) that appears due to the correlation between the two innovation sequences of the predictive regression and the autoregressive equation, parameter estimation is achieved using the IVX filtration (see, \cite{phillipsmagdal2009econometric} and \cite{kostakis2015Robust}), which is robust to  the unknown persistence properties of regressors. In particular, the IVX estimation approach is implemented to both stages of the procedure, thereby ensuring identification and joint estimation of the risk measures of $\mathsf{VaR}$ and $\mathsf{CoVaR}$\footnote{Although statistical comparisons between the assumption of stationary regressors against regressors generated via a LUR parametrization can be constructed more formally using goodness-of-fit testing (e.g., see \cite{chan2013marked}); we leave these considerations for future research. On the other hand, more recently \cite{hoga2023estimation} propose a framework for the estimation risk in extreme systemic risk forecasts for stationary time series data.}.

\newpage

We focus on the IVX estimators in each estimation stage, since it is well-known to be robust to the abstract degree of persistence (e.g., see \cite{lee2016predictive}). As a result, the doubly corrected IVX estimator (which is the IVX estimator obtained from the second step estimation procedure), is the main parameter of interest in terms of statistical inference and asymptotic theory. Specifically, we establish the asymptotic properties of the doubly corrected IVX estimator which verify the mixed Gaussianity property of the limiting distribution regardless of the degree of persistence in both estimation stages. Moreover, we consider a suitable correction to the expression of the asymptotic variance-covariance matrix in the second estimation step, which is adjusted in order to account for the first-step estimation error that produces the generated regressor under nonstationarity. 

Regarding the statistical properties of parameters obtained from the first-stage procedure those are employed for the $\mathsf{VaR}$ risk measure forecast and in practice parameter consistency is not violated (see, \cite{AB2016covar}) regardless of the LUR parametrization. In addition, the particular generated regressor that corresponds to fitted values of a stationary regressand can be treated as the stationary component of the vector of regressors used during the second stage procedure ($\mathsf{CoVaR}$). However, the generated regressors are included in the QPR model of the second estimation step with some sampling error, which introduces additional noise into the asymptotic variance-covariance matrix of the model coefficients. Therefore, the sampling error from the first stage contaminates the second stage estimation procedure (see, \cite{chen2021quantile}). Furthermore, during the second-stage ($\mathsf{CoVaR}$ estimation), model estimates are obtained using the IVX estimator in which case the generated regressor is exogenously determined, this implies that it can only be assumed to be stationary (although is modelled as the same lag as the regressand). As a result, the usual way of calculating the QR variance-covariance matrix fails to account for the additional source of error and an adjustment is required. To sum-up, the generated regressor is a proxy of the estimated $\mathsf{CoVaR}$ when the nonstationary properties of regressors are included, which implies that this proxy is treated as a nonstochastic stationary regressor in a quantile predictive regression model with nonstationary regressors. Thus, accounting for the presence of generated regressor bias, inference is robustify against both the abstract degree of persistence that these nonstationary regressors exhibit as well as to the presence of any generated regressor bias.

Recall that, the conditional quantile function $F^{-1}_{ \textcolor{blue}{ y_{(j)t} } } \left( \uptau \big| \boldsymbol{x}_{t-1} , \textcolor{blue}{ y_{(i)t} } \right)$ is the $\mathsf{VaR}_{ (j) t} \left( \uptau \right)$ conditional on $\boldsymbol{x}_{t-1}$ and $\textcolor{blue}{ y_{(i)t} }$, where $F(\cdot)$ is the conditional distribution function and $F^{-1} (\cdot)$ is the conditional quantile function of the underline distribution. Then, by conditioning on $\textcolor{blue}{ y_{(i)t} } =\mathsf{VaR}_{ (i) t} \left( \uptau \right)$, we obtain the $\mathsf{CoVaR}_{ (j,i) t} \left( \uptau \right)$ defined as (see, \cite{AB2016covar} and \cite{hardle2016tenet})
\begin{align}
\label{A}
\mathsf{CoVaR}_{ (j,i) t} \left( \uptau \right) 
&:= 
\underset{  \mathsf{VaR}_{ (j) t} \left( \uptau \right) }{ \mathsf{inf} } \ \left\{ \mathbb{P} \bigg( \textcolor{blue}{ y_{(j)t} } \bigg| \left\{ \boldsymbol{x}_{t-1} ,  \textcolor{blue}{ y_{(i)t} } = \mathsf{VaR}_{ (i) t} \left( \uptau \right) \right\} \bigg) \geq \uptau \right\} 
\end{align}
Equivalently, it holds that 
\begin{align}
\label{B}
\mathsf{CoVaR}_{ (j,i) t} \left( \uptau \right) \equiv F^{-1}_{ \textcolor{blue}{ y_{(j)t} } } \left( \uptau \big| \boldsymbol{x}_{t-1} , \textcolor{blue}{ y_{(i)t} } =  \mathsf{VaR}_{ (i) t} \left( \uptau \right) \right), \ \ \  \uptau \in (0,1).
\end{align}

\newpage 

Thus, we focus on estimating systemic risk measures conditioned on time series nonstationarity (persistent data) using the LUR parametrization. The nonstationary quantile regression specifications are employed for forecasting the pair $( \mathsf{VaR}, \mathsf{CoVaR} )$, based on the definitions of \cite{AB2016covar} and the statistical elicitability property of these risk measures (see, \cite{patton2019dynamic}, \cite{frongillo2021elicitation} and \cite{fissler2016higher}). Although, the loss function which characterizes M-estimators allows to obtain model estimates, thereby directly affecting the robustness of forecasts, we do not consider the estimation risk in systemic risk forecasts (see,   \cite{dimitriadis2022characterizing}, \cite{xu2022prediction}, \cite{hoga2023estimation} and \cite{katsouris2023statistical}).

\subsection{Main contributions}

\begin{itemize}

\item[\textit{(i)}.]   Our first main contribution to the literature is to introduce dynamic co-quantile predictive regression models with nonstationary regressors ($\mathsf{CoQR-LUR}$). Our generalized class of estimators is employed to jointly estimate the vector $(\mathsf{VaR}, \mathsf{CoVaR})$ by incorporating the persistence properties of regressors. Our proposed general class of models nests both the predictive co-quantile regressions in the spirit of \cite{AB2016covar} as well as the quantile predictive regressions with persistence covariates in the spirit of \cite{lee2016predictive}.      

\item[\textit{(ii)}.] Our second main contribution is to propose an estimator for the model parameters and derive its large sample properties. The main technical difficulty to overcome in developing related asymptotic theory, is that the joint estimation of the risk measures of $( \mathsf{VaR}, \mathsf{CoVaR} )$ fails to be elicitable, such that no real-valued scoring function exists that is uniquely minimized by the true values of these risk measures (see, \cite{dimitriadis2022dynamic}). This pitfall suggests the implementation of a two-step M-estimator  (see, also \cite{martins2018nonparametric}, \cite{patton2019dynamic} and \cite{fang2023simple}). Our proposed estimation approach extends to nonstationary data which is more general than current estimation methods. 

\item[\textit{(iii)}.] We focus on showing consistency and Asymptotic Mixed Gaussianity of the proposed doubly IVX robust estimator and propose a valid inference procedure based on consistent estimation of the asymptotic variance-covariance matrix. Our asymptotic theory analysis includes the use of the standard local-to-unity asymptotics found in relevant studies (see, \cite{lee2016predictive} and \cite{cai2023new}), but also considers the adjustment for the variance-covariance matrix to account for the estimation uncertainty of standard errors during the first stage procedure\footnote{It is worth pointing out that formal testing methods for determining the joint statistical significance of these nearly integrated regressors as predictors for systemic risk is a novel aspect in the literature to the best of our knowledge when nonstationary predictive regressions include a generated covariate.}.    

\end{itemize}
The use of endogenously generated instruments (IVX filter\footnote{The ideal IV has to meet two requirements: it is correlated with the endogenous regressor via an explainable and validated relationship (relevance relationship), yet uncorrelated with the structural error (exclusion restriction).}) is employed for the purpose of filtering out the unknown form of persistence in the regressors of each quantile predictive regression used in the two-stage estimation procedure.

\newpage 

Regarding identification issues of the risk measure pair $( \mathsf{VaR}, \mathsf{CoVaR})$ are discussed in several studies such as \cite{AB2016covar}, \cite{hardle2016tenet}, \cite{patton2019dynamic}. Moreover,  \cite{fissler2023backtesting} mention that although there is no real-valued scoring function associated with risk measure pairs and propose to consider a $\mathbb{R}^2-$valued scoring function in the closely related context of forecast evaluation. Nevertheless, our estimation method does not rely on the score function estimation approach, thereby the statistical elicitability property of the risk measure pair facilitates econometric estimation.  We also avoid using copula correction methods (see,  \cite{yang2022addressing} among others), since within the setting of nonstationary predictive regression models, the IVX filtration of \cite{phillipsmagdal2009econometric} has been found to be robust to abstract degree of persistence. In addition to the proposed two-stage algorithm for the joint estimation of the risk measure pair, as well as  asymptotic theory analysis of estimators and their covariance matrices, our framework allows to study  statistical inference issues which we leave for future research\footnote{The related statistical inference aspects we refer to here is the \textit{subvector inference problem} which is commonly examined in the linear instrumental variables regression model. In that case, our objective is to develop a powerful test such as the asymptotic null rejection probability is controlled uniformly over a parameter space. The inference problem  consists of making inference on the coefficients of structural models, but in the context of systemic risk modelling. Such statistical problems usually give attention to aspects such as developing robust methods in the presence of weak instruments or powerful subvector testing in IV regressions under conditional heteroscedasticity; we give attention on the simultaneous modelling of risk measures which involve the use of a generated covariate. Therefore,  drawing inference on individual coefficients of the parameter vector of the structural equation for estimating the $\mathsf{CoVaR}$ risk measure, motivates us to propose a subvector testing approach which is based on marginal conditional regression estimation and testing with nonstationary regressors (see, \cite{hsiao1997statistical} and \cite{de2001nonlinear}). }.

We propose a practical two-stage estimation procedure\footnote{We study a broad class of semiparametric problems, where a doubly robust estimating function $h_{\uptau}( U; \theta, \theta_1, \theta_2 )$, where $\uptau \in (0,1)$ for the parameter of interest $\theta$ is available as follows. Consider two working models denoted with $g_1( \boldsymbol{x}_{t-1} ; \theta_1 )$ and $g_{2|1} ( \boldsymbol{x}_{t-1} ; \theta_2 )$ for unknown functions $y_{1t} = g_1^*( \boldsymbol{x}_{t-1} )$ and $y_{2t} = g_2^*( \boldsymbol{x}_{t-1} )$. Furthermore, the estimating function (in our case it represents the tail estimates for the $\mathsf{CoVaR}$ risk measure), $h_{\uptau}( . ; . )$ is unbiased such that $\mathbb{E} \bigg[ h_{\uptau} \bigg( U; \theta, \theta_1, \theta_2 \bigg) \bigg] = 0$, when $\theta$ is set to the true value $\theta^{*}$, and either $\theta_1$ or $\theta_2$, but not necessarily both, is set to the true value $\theta_1^{*}$ or $\theta_2^{*}$ defined respectively such that $g_1^*( \boldsymbol{x}_{t-1} ) \equiv g_1( \boldsymbol{x}_{t-1} )$ or  $g_2^*( \boldsymbol{x}_{t-1} ) \equiv g_2( \boldsymbol{x}_{t-1} )$ if model $g_1 ( \boldsymbol{x}_{t-1} ; \theta_1 )$ or $g_2 ( \boldsymbol{x}_{t-1} ; \theta_2 )$ is correctly specified. In general, doubly robust estimation using the estimating function $h_{\uptau}( U; \theta, \theta_1, \theta_2 )$ consists of two stages: some estimators $\left( \hat{\theta}_1, \hat{\theta}_2 \right)$ are first defined, and then $\hat{\theta}$ is defined by solving the estimating equation $\tilde{\mathbb{E} } \big[ h_{\uptau} \big( U; \theta, \hat{\theta}_1, \hat{\theta}_2 \big) \big] = 0$, where $\tilde{\mathbb{E} } [ \cdot ]$ denotes a sample average.} for estimating the model parameters of the quantile predictive regression with nonstationary regressors. Then, the second step employs the generated regressors in the form of fitted values, along with a different set of nonstationary regressors to a quantile predictive regression model. As a result, the doubly IVX corrected estimator, is the IVX-based estimator corresponding to the second step estimation procedure and is the main parameter of interest in terms of statistical inference and limit theory. Since the generated regressor is assumed to be estimated using a different functional form as well as nonstationary regressors, then any dependence structure between the generated regressors is captured by the variance covariance matrix of the second-step estimation procedure. Within our econometric environment the proposed two-stage (semiparametric) estimation method, depends on possibly nonstationary regressors, therefore relevant parameter restrictions on $\vartheta_1$ and $\vartheta_2$ (see, also \cite{katsouris2021optimal, katsouris2023statistical}), which correspond to the "working" models, are necessary to ensure robust statistical inference. 

\newpage

For notational convenience consider the following two functional forms: 
\begin{align}
g^{*}_1 ( \boldsymbol{x}_{t-1}  ) 
&\equiv
g_1 \big( \boldsymbol{x}_{t-1}  ; \boldsymbol{\beta}_1 \big) = \psi_{(1) \uptau } \left( \boldsymbol{\beta}_1^{\top} \xi ( \boldsymbol{x}_{t-1} ) \right),
\\
g^{*}_2 ( \boldsymbol{x}_{t-1}  ) 
&\equiv
g_2 \big( \boldsymbol{x}_{t-1}  ; \boldsymbol{\beta}_2 \big) = \psi_{(2) \uptau } \left( \boldsymbol{\beta}_2^{\top} \xi ( \boldsymbol{x}_{t-1} ) \right),
\end{align}
where $\xi ( \boldsymbol{x}_{t-1} )$ is $p \times 1$ vector of known functions such as $\xi ( \boldsymbol{x}_{t-1} ) = \big( 1, \boldsymbol{x}_{t-1}^{\top} \big)$, and $\boldsymbol{\beta}_1$ and $\boldsymbol{\beta}_2$ are vectors of unknown coefficients, which correspond to joint estimation as the two models in \cite{AB2016covar}.  Under time series nonstationarity, the two parameter vectors are estimated via the IVX filtration of \cite{phillipsmagdal2009econometric}. Our implementation is, however, a plug-in estimation approach and a novel methodology to existing approaches discussed in the literature. Therefore, we consider the asymptotic theory analysis for $\big( \hat{\boldsymbol{\beta}}_1, \hat{\boldsymbol{\beta}}_2 \big)$ and $\widehat{\vartheta} \big( \hat{\boldsymbol{\beta}}_1, \hat{\boldsymbol{\beta}}_2 \big)$,  where $\boldsymbol{\beta}_1$ and $\boldsymbol{\beta}_2$ are fixed-dimensional as the sample size $n$ grows. In particular, we assume that $\hat{\boldsymbol{\beta}}_1$ is consistent for $\boldsymbol{\beta}^{*}_1$ when the first working model is correctly specified, and $\hat{\boldsymbol{\beta}}_2$ is consistent for $\boldsymbol{\beta}^{*}_2$ if the second working model is correctly specified (absence of functional form misspecification). 

Assume the existence of any two pair of random variables $\left\{ y_{1t}, y_{2t} \right\}$ and a set of predictors $\boldsymbol{x}_t$. Another way of viewing our econometric environment is to consider the example of two-countries (or firms across two countries) where we have a balanced set of $p$ nonstationary macroeconomic variables for both economies (or firms), respectively. In particular, let $\boldsymbol{x}_{1t}$ be the $( p \times 1)$ vector of domestic variables and $\boldsymbol{x}_{2t}$ be the $( p \times 1)$ vector of the same variables for the foreign country where each of these possible nonstationary processes have the following LUR parametrizations
\begin{align}
\boldsymbol{x}_{1t} &= \boldsymbol{R}_{1n} \boldsymbol{x}_{1t-1} + \boldsymbol{u}_{1t}, \ \boldsymbol{R}_{1n} = \left( \boldsymbol{I}_p - \frac{ \boldsymbol{C}_{1p} }{n}   \right).
\\
\boldsymbol{x}_{2t} &= \boldsymbol{R}_{2n} \boldsymbol{x}_{2t-1} + \boldsymbol{u}_{2t}, \ \boldsymbol{R}_{2n} = \left( \boldsymbol{I}_p - \frac{ \boldsymbol{C}_{2p} }{n}   \right).
\end{align} 
Moreover, additionally to these nonstationary regressors that are equation-specific, a generated covariate is employed as proxy for the risk measure of $\mathsf{VaR}$ that represents a quantile risk event of the  macroeconomic conditions of the foreign economy when predicting the current time conditions of the home economy. Our estimation approach combines both the IVX estimator of \cite{phillipsmagdal2009econometric} (to filter out persistence present in the nonstationary regressors $\boldsymbol{x}_{1t}$ and  $\boldsymbol{x}_{2t}$), as well as inspired by the residual augmentation approach of \cite{demetrescu2020residual} we self-instrument the proxy variable (generated covariate) under time series nonstationarity (see, also \cite{demetrescu2014enhancing} and \cite{breitung2015instrumental}).  As a result, when developing the asymptotic theory of the (long-run) quantile-dependent parameter we use a self-instrumentation approach for the generated regressors that correspond to the systemic risk proxy of each economy (as the joint estimation approach of \cite{AB2016covar} and \cite{hardle2016tenet}).

\newpage

Estimation of the $\mathsf{CoVaR}$ risk measure is conditioned on the estimated $\mathsf{VaR}$, and relies on a quantile predictive regression, with both nearly integrated regressors and a generated regressor (proxy for estimated $\mathsf{VaR}$ under nonstationarity). Therefore, in the presence of this additional covariate, there are more regressors than available instruments. To overcome this issue, \cite{demetrescu2014enhancing} propose a framework for inference in an augmented predictive regression. In that case, the author propose to instrument that additional covariate with the same process itself. We follow a similar approach in our setting although with major necessary modifications. 

\subsection{Estimation Algorithm}

Our proposed two-stage estimation procedure is summarized as below:

\begin{itemize}

\item[\textbf{Step 1.}] (Estimation of $\mathsf{VaR}$ under time series nonstationarity). Let $\boldsymbol{\beta}_1 (\uptau)$  be the parameter vector of interest (without IVX filtration yet). Then, the QR-estimator is (see, \cite{lee2016predictive}) 
 \begin{align}
\widehat{ \boldsymbol{\beta} }_1 (\uptau) 
= 
\underset{ \beta \in \mathbb{R}^p }{ \mathsf{arg \ min} }  \sum_{i=1}^n \rho_{\uptau} \bigg( y_{1t} - \boldsymbol{x}_{1t-1}^{\prime} \boldsymbol{\beta}_1 \bigg).
\end{align} 
where the check function is defined as $\rho_{\uptau} ( \mathsf{u}  ) :=  \mathsf{u} \big[ \uptau - \mathbf{1} \left\{ \mathsf{u} \geq 0 \right\} \big]$ and $\boldsymbol{x}_t$ are LUR regressors. 

During this step in practice we replace the QR-estimator with the QR-IVX estimator from \cite{lee2016predictive} due to the presence of nonstationary regressors. Then, the estimated one-period ahead risk measure of $\mathsf{VaR}$ is obtained, denoted by $\hat{y}_{1t}^{\mathsf{VaR}}$. In the case of $N$ cross-sectional units all these generated covariates (under time series nontationarity), shall be computed\footnote{Note that the main idea of modelling systemic risk based in this approach is proposed by \cite{katsouris2021optimal} and \cite{katsouris2023statistical} which is also related to various studies in the literature of financial connectedness (see, \cite{billio2012econometric}, \cite{diebold2012better} and \cite{diebold2014network} among others.}.

\item[\textbf{Step 2.}] (Estimation of $\mathsf{CoVaR}$ under time series nonstationarity). Compute the unknown parameter vector $\boldsymbol{\beta}^{\star} _2 (\uptau)$  from the following QR model 
\begin{align}
\widehat{ \boldsymbol{\beta} }_2^{\star} (\uptau) 
= 
\underset{ \boldsymbol{\beta}_2^{\star} \in \mathbb{R}^{p + 1} }{ \mathsf{arg \ min} }  \sum_{i=1}^n \rho_{\uptau} \bigg( y_{2t} - \tilde{\boldsymbol{x}}_{2t-1}^{\prime} \boldsymbol{\beta}_2^{\star} \bigg), \ \ \ \tilde{\boldsymbol{x}}_{2t-1} = \big[ \boldsymbol{x}_{2t-1} \ \ \ \widehat{y}^{\mathsf{VaR}}_{1t}  \big].
\end{align}
where $\boldsymbol{\beta}_2^{\star}$ corresponds to the parameter vector of the nonstationary regressors of the second-stage as well as the unknown parameter that correspond to the fitted-values of the $\mathsf{VaR}$\footnote{In the framework of \cite{hardle2016tenet}, this coefficient corresponds to the unknown model parameter $\hat{\beta}_{j|i}$ in their expression (6). See also expressions (7) and (8) of \cite{hardle2016tenet} that correspond to the econometric specifications used to obtain the estimates of the risk measure pair under the assumption of stationary regressors.} that have been estimated during the first-stage estimation procedure. 
\end{itemize}
Recall that the conditional quantile can be consistently estimated  based on the quantile-dependent model parameter $\widehat{ \boldsymbol{\beta} }_2^{\star} (\uptau)$ such that $\widehat{Q}_{\uptau} (  \uptau |  \boldsymbol{x}_{2t-1},  \widehat{y}^{\mathsf{VaR}}_{1t} ) = \tilde{\boldsymbol{x}}_{2t-1}^{\prime
}  \widehat{ \boldsymbol{\beta} }_2^{\star} (\uptau)$ (see, \cite{koenker1978regression}). In a similar spirit the corresponding IVX estimators are obtained. 

\newpage

Therefore, one uses the estimates of $\boldsymbol{\theta}_j$, denoted by $\widehat{\boldsymbol{\theta}}_j$, in the first step, to obtain the generated regressor $\widehat{\boldsymbol{x}}_t$. Our aim is to demonstrate that our proposed methodology achieve rate and model double robustness simultaneously, provided that the parameter $\boldsymbol{\theta}$ satisfies certain regularity conditions. In the current study the doubly IVX corrected estimation approach corresponds to conditional quantile specification forms for nonstationary data. 

The remained of this paper is organized as follows. In Section \ref{Section2}, we introduce the econometric model and the main assumptions of our framework. In  Section \ref{Section3}  we present our proposed doubly corrected estimation method in the case of conditional mean functional forms while in Section  \ref{Section4} discusses the conditional quantile case which corresponds to the estimation of the risk measures under time series nonstationarity. We conclude in Section \ref{Section5}. 

Throughout the paper, for any vector $\boldsymbol{a}$, we denote with $\norm{ \boldsymbol{a} }_1 = \sum_j | a_j |$ and with $\norm{ \boldsymbol{a} }_2 = \sqrt{ \sum_j | a_j |^2 }$. For any real arbitrary matrix $\boldsymbol{A}$, its Frobenius norm is defined by $\norm{ \boldsymbol{A} } = \sqrt{ \mathsf{trace} ( \boldsymbol{A}^{\prime} \boldsymbol{A} ) }$.  The spectral norm of $\boldsymbol{A} $ is denoted by $\norm{ \boldsymbol{A}  }_2$, such that, $\norm{ \boldsymbol{A}  }_2 = \sqrt{ \lambda_1 ( \boldsymbol{A}^{\prime} \boldsymbol{A}  )  }$, is maximum column sum norm is $\norm{ \boldsymbol{A} }_1 = \norm{ \boldsymbol{A} }_1 = \mathsf{max}_{ 1 \leq j \leq n}  \sum_{i = 1}^m | A_{ij} |$ and its maximum row sum norm is $\norm{ \boldsymbol{A} }_{ \infty } = \norm{ \boldsymbol{A} }_{ \infty } = \mathsf{max}_{ 1 \leq i \leq n}  \sum_{i = 1}^m | A_{ij} |$.  The operator $\overset{P}{\to}$ denotes convergence in probability, and $\overset{D}{\to}$ denotes convergence in distribution.


\begin{wrap}

\subsection{Discussion of Alternative Estimation Procedures under stationarity}

Consider the pair $\left\{ (X_t, Y_t)^{\prime} \right\}_{ t \in\mathbb{N} }$, where $X_t$ can be a possibly nonstationary time series process. Based on the definitions of the $\mathsf{VaR}$ and $\mathsf{CoVaR}$, we follow the approach presented by \cite{AB2016covar}. In particular, in that case the stress event that is considered for the estimation of $\mathsf{CoVaR}$ is that the loss of the reference position equals that of its $\mathsf{VaR}$, such that, $\left\{ X_t = \mathsf{VaR}_t   \right\}$. In addition, we assume that both these risk measures are estimated at the same quantile level $\uptau$. Following \cite{dimitriadis2022dynamic}, we assume that the econometric model - that is, the quantile predictive regression, is correctly specified in the sense that there exist true parameter values $\boldsymbol{\theta}_0^v \in \boldsymbol{\Theta}^{v} \in \mathbb{R}^p$ and $\boldsymbol{\theta}_0^{c} \in \boldsymbol{\Theta}^{n} \in \mathbb{R}^q$ such that the following expression holds
\begin{align}
\begin{pmatrix}
\mathsf{VaR}_t (\uptau)
\\
\mathsf{CoVaR}_t (\uptau)
\end{pmatrix}    
=
\begin{pmatrix}
v_t ( \boldsymbol{\theta}_0^v )  
\\
v_t ( \boldsymbol{\theta}_0^c ) 
\end{pmatrix}
\end{align}
Recently, \cite{fissler2023backtesting} proposed to consider a $\mathbb{R}^2-$valued scoring function in the closely related context of forecast evaluation. In particular, the authors show that under some regularity conditions, the expectation of the $\mathbb{R}^2-$valued scoring function given by 
\begin{align*}
\boldsymbol{S} 
\left(  
\begin{pmatrix}
v
\\
c
\end{pmatrix}, 
\begin{pmatrix}
x
\\
y
\end{pmatrix}
\right)    
= 
\begin{pmatrix}
S^{ VaR } (v, x)
\\
S^{ CoVaR } \big( (v, x)^{\prime}, (x,y)^{\prime} \big) 
\end{pmatrix}
=
\begin{pmatrix}
\big[ \mathbf{1} \left\{  x \leq v \right\} - \beta \big] ( v - x )
\\
\mathbf{1} \left\{  x > v \right\} \big[  \mathbf{1} \left\{  y \leq c \right\}  - \alpha \big] (c - x)
\end{pmatrix}
\end{align*}
is minimized by the true $\mathsf{VaR}$ and $\mathsf{CoVaR}$ with respect to the lexicographic order. Then, the model setting is parametric with respect to the estimation of the parameter vector which is obtained by fitting a quantile predictive regression model. In particular, we consider a two stage M estimator which during the first stage we obtain the parameter vector $\boldsymbol{\theta}_0^v$ while during the second stage we obtain the parameter vector $\boldsymbol{\theta}_0^c$. During the first-stage of the procedure, the parameter vector that corresponds to the estimation of the $\mathsf{VaR}$ such that
\begin{align}
\label{function1}
\widehat{\boldsymbol{\theta}}_n^v = \underset{ \boldsymbol{\theta}^v \in \mathbf{\Theta}^v }{ \mathsf{arg min}  } \ \frac{1}{n} \sum_{t=1}^n S^{ \mathsf{VaR} } \big( v_t( \boldsymbol{\theta}^v ), X_{t-1} \big),
\end{align}
where $v_t( \boldsymbol{\theta}^v )$ is the quantile predictive regression estimator for some $\mathsf{VaR}$ parameter space $\boldsymbol{\Theta}^v \in \mathbb{R}^p$ that corresponds to a $p-$dimensional stationary regressors. During the second stage of the estimation procedure we obtain the parameter vector $\boldsymbol{\theta}_0^c$ as
\begin{align}
\label{function2}
\widehat{\boldsymbol{\theta}}_n^c = \underset{ \boldsymbol{\theta}^c \in \mathbf{\Theta}^c }{ \mathsf{arg min}  } \ \frac{1}{n} \sum_{t=1}^n S^{ \mathsf{CoVaR} } \left( v_t( \widehat{\boldsymbol{\theta}}_n^v ), c_t ( \boldsymbol{\theta}^c ) \big), \big( X_{t-1}, Y_t \big) \right) 
\end{align}
for some $\mathsf{CoVaR}$ parameter space $\boldsymbol{\Theta}^c \in \mathbb{R}^q$ of $q-$dimensional parameter vector of stationary regressors and the generated covariate. For the proposed two-stage estimation and $M-$estimator to be feasible, the requirement is that the $\mathsf{VaR}$ evolution does not depend on the $\mathsf{CoVaR}$ model. 


\subsubsection{Asymptotic Properties of Joint M-Estimators}

\paragraph{Consistency.} 

An econometrician is first interested in showing consistency of the proposed two-stage M-estimators $\big( \widehat{\boldsymbol{\theta}}_n^v , \widehat{\boldsymbol{\theta}}_n^c \big)$ with stationary regressors. In particular, in a stationary time series setting, \cite{fissler2023backtesting} shows that as $n \to \infty$, $
\widehat{\boldsymbol{\theta}}_n^v \overset{\mathbb{P}}{\to} \boldsymbol{\theta}_0^v$ and  $\widehat{\boldsymbol{\theta}}_n^c \overset{\mathbb{P}}{\to} \boldsymbol{\theta}_0^c$, where $\widehat{\boldsymbol{\theta}}_n^v$ and $\widehat{\boldsymbol{\theta}}_n^c$ are defined by the optimization functions \eqref{function1} and \eqref{function2} respectively.   

\begin{example}
Consider the quantile predictive regression model as below
\begin{align}
y_t = \alpha (\uptau) + \beta (\uptau) x_{t-1} + u_{t} (\uptau) \ \ \text{or equivalently} \ \ \ Q_{  y_t }  ( \tau \mid \mathcal{F}_{t-1} )  = \alpha (\uptau) + \beta (\uptau) x_{t-1},       
\end{align}
such that $\mathbb{P} \big( y_t \leq Q_{ y_t } ( \tau \mid \mathcal{F}_{t-1} ) \mid \mathcal{F}_{t-1} \big) = \uptau \in (0,1)$, where $\mathcal{F}_{t-1}$ is the information available at time $t-1$. Conditional (systemic) risk measure forecasts are more informative than unconditional ones. In particular, conditional risk measures are based on the conditional distribution of $( X_t, Y_t )$, which is given by the expression $F_{ (X_t, Y_t ) \mid \mathcal{F}_{t-1} } := \mathbb{P} \big( X_t \leq x; Y_t \leq y \mid \mathcal{F}_{t-1} \big)$.   Moreover, assume that the $x_{t-1}$ is a predictor which has the following possibly nonstationary autoregressive representation $x_t = \mu + \rho_n x_{t-1} + v_t$. Then the QR-induced regressor errors $\psi_{\uptau} (u_t ( \tau ) )$ are contemporaneously correlated with the innovations of unit root sequences and inference can be robustify\footnote{Although our proposed estimation methodology considers two different models to obtain forecasts for the $\mathsf{VaR}$ and $\mathsf{CoVaR}$, these can not be considered as standalone risk measures. Specifically, although we impose the assumption that the evolution of the $\mathsf{VaR}$, doesn't affect the estimation of the $\mathsf{CoVaR}$, the estimation of $\mathsf{CoVaR}$ requires pre-estimation of the generated regressor during the first-stage, otherwise these risk measures are neither identifiable not elicitable (see, \cite{patton2019dynamic}, \cite{hardle2016tenet} and \cite{AB2016covar}).} using the IVX estimator (see, \cite{lee2016predictive}, \cite{fan2019predictive} and \cite{fan2023predictive}).  
\end{example}

\begin{remark}
Our main focus in not whether forecasts are misspecified or not (e.g., see  \cite{dimitriadis2021forecast} and \cite{dimitriadis2023encompassing}) but whether these can be jointly and consistently estimated under the assumption that the distribution function of the regressors of the model are generated as nearly integrated processes. Moreover, elicitability and identifiability is important for comparing and evaluating (risk) forecasts in practise. We also assume the absence of forecasts misspecification and we focus on the joint identification and estimation of these risk measures under the presence of regressors nonstationarity in the sense that the sequence $(X_t)_{ t \in \mathbb{Z} }$ is not necessarily stationary. Thus, the asymptotic theory in \cite{fissler2023backtesting} is established byroving the consistency of the parameter vector that corresponds to the VaR, under the assumption that the model is correctly specified. 
\begin{align}
Q_0^v ( \boldsymbol{\theta}^v ) := \mathbb{E} \left[ \frac{1}{n} \sum_{t=1}^n S^{ \mathsf{VaR} } \big( v_t( \boldsymbol{\theta}^v  ), X_{t-1} \big) \right]    
\end{align}
Under the assumption of correct model specification, it holds that the conditional functional given by $Q_0^v ( \boldsymbol{\theta}^v )$ is uniquely minimized at the optimal value of the $\mathsf{VaR} ( X_t  \mathcal{F}_{t-1} )$. Therefore, one needs to consider whether the second-stage M-estimator is consistently estimated conditional on the first-stage M-estimator that corresponds to the $\mathsf{VaR}$. The conditional first-stage consistency property implies that we can rewrite the following expression conditional on the estimator given in the previous step, such that (see, \cite{fissler2023backtesting})
\begin{align*}
\displaystyle \widehat{Q}_n =
\frac{1}{n} \sum_{t=1}^n S^{\mathsf{CoVaR} } \big(  \big( v_t( \widehat{\boldsymbol{\theta}}^v ), c_t ( \boldsymbol{\theta}^c ) ,  ( Y_t,  X_{t-1} \big)^{\prime} \big)   \big).
\end{align*}
\end{remark}

\subsubsection{Asymptotic Mixed Gaussianity} 

The second main result we focus on, is to derive the asymptotic mixed Gaussianity of the second stage M-estimator under regressors nonstationarity. This result is summarized in the following theorem. In particular, the first part of the theorem below obtain a Bahadur representation for the unknown parameter vector that corresponds to $\mathsf{CoVaR}$. A key point here is to identify the correct convergence rate, which is likely to not be the standard non-parametric rate due to the existence of nearly integrated regressors as well as the generated covariate. 

\begin{theorem}
The proposed two-stage estimation method of \cite{fissler2023backtesting} implies that
\begin{align}
\sqrt{n} \left(  \widehat{\boldsymbol{\theta}}_n^v - \boldsymbol{\theta}_0^v  \right)  = \left[ \boldsymbol{\Lambda}^{-1}  + o_p(1) \right] \left[  - \frac{1}{ \sqrt{n} } \sum_{t=1}^n \boldsymbol{v}_t ( \boldsymbol{\theta}_0^v ) + o_p(1)  \right]  
\end{align}
By applying Slutsky's theorem the limit is more likely to have mixed Gaussian distribution with a stochastic covariance matrix, such that,  $\sqrt{n}  \left(  \widehat{\boldsymbol{\theta}}_n^v - \boldsymbol{\theta}_0^v  \right) \overset{d}{\to} \mathcal{N} \left( 0, \boldsymbol{\Lambda}^{-1}  \boldsymbol{V}  \boldsymbol{\Lambda}^{-1}  \right)$ as $n \to \infty$.
\end{theorem}

\begin{remark}
All the above covariance matrices have to be estimated analytically and are likely to affect the inference step. Notice that for simplicity, one can consider that both risk measures are estimated around the same quantile level which is kept fixed. However, due to the various challenges one needs to tackle when considering the \textit{lexicographic} approach proposed by \cite{dimitriadis2022dynamic}, under regressors nonstationarity, we decide to follow our original idea of constructing the pair of risk measures  $(\mathsf{VaR}, \mathsf{CoVaR})$ using the check function when estimating the quantile predictive regression models in a sequential manner (proposed two-stage estimation procedure).
\end{remark}

\end{wrap}


\section{Econometric Model and Assumptions}
\label{Section2}

\subsection{Main Assumptions}

Consider the linear predictive regression model formulated as below
\begin{align}
y_t 
&= 
\mu + \boldsymbol{\beta}^{\prime} \boldsymbol{x}_{t-1} + u_t, \ \ \ 1 \leq t \leq n 
\\
\boldsymbol{x}_{t}
&= 
\boldsymbol{R}_n \boldsymbol{x}_{t-1} + \boldsymbol{v}_t
\end{align}
where $\boldsymbol{x}_{t} \in \mathbb{R}^p$ is a $p-$dimensional vector and $\boldsymbol{R}_n = \left( \boldsymbol{I}_p - \frac{ \boldsymbol{C}_p }{n^{\gamma}}   \right)$ with $\gamma = 1$.

We begin by discussing the predictive regression with a conditional mean functional form,  but for the purpose of estimating the risk measures pair $( \mathsf{VaR}, \mathsf{CoVaR} )$ a conditional quantile specification form is necessary. The asymptotic behaviour of estimators under moderate deviations from unity (e.g., on the explosive side) is beyond the scope of our study (e.g., see, \cite{phillips2007limit}, \cite{buchmann2007asymptotic}). Several studies consider estimation methods for unified inference in predictive regression models such as \cite{magdalinos2020least} and \cite{magdalinos2022uniform}. Furthermore, our estimation method differs in several ways to various respectful school of thoughts and approaches presented in the literature such as the two-step estimator of  \cite{francq2015risk} (see also \cite{beutner2024residual} and \cite{francq2023local}); the research perspective we follow for the development of our econometric environment is concisely explained in Section 1 (see, \cite{AB2016covar} and \cite{hardle2016tenet}).

\newpage

\begin{assumption}
\label{Assumption1}
Let $\boldsymbol{\epsilon}_t = \left( u_t, \boldsymbol{e}^{\prime}_t   \right)^{\prime}$ be an $\mathbb{R}^{p + 1}-$valued martingale difference sequence with respect to the filtration $\mathcal{F}_t = \sigma \big( \boldsymbol{\epsilon}_t, \boldsymbol{\epsilon}_{t-1}, ... \big)$ with $\mathbb{E} \big[ \boldsymbol{\epsilon}_t \boldsymbol{\epsilon}_t^{\prime}  \mid \mathcal{F}_{t-1} \big] = \boldsymbol{\Sigma}_{\epsilon} > 0$.  Let $\boldsymbol{v}_t$ be an $\mathbb{R}^p-$valued stationary linear process such that 
\begin{align*}
\boldsymbol{v}_t = \sum_{j=0}^{\infty} \boldsymbol{F}_j \boldsymbol{e}_{t-j}, \ \ \ \text{with} \ \ \ \sum_{j=0}^{\infty} \norm{ \boldsymbol{F}_j  } < \infty,
\end{align*}
where $\boldsymbol{F}_j$ is a sequence of constant matrices such that $\boldsymbol{I}_0 = \boldsymbol{I}_p$, $\displaystyle \sum_{j=0}^{\infty} \boldsymbol{F}_j$ has full rank.
\end{assumption}
Borrowing from the literature in the case when $\sigma_{uv} \neq 0$ (which is the case we consider), the exact OLS bias of $\hat{\beta}$ is computed from the predictive regression model, $\mathbb{E} \left[ \hat{\beta} - \beta \right] = \delta \mathbb{E} \left[ \hat{\rho} - \rho \right]$ where $\hat{\rho}$ is the OLS estimate of $\rho$ and $\delta := \sigma_{uv} / \sigma_v^2$ is the slope coefficient of in a regression of $u_t$ on $v_t$. Furthermore, since $\hat{\rho}$ is known to be downward biased in small-samples, and $( u_t, v_t )^{\prime}$ are typically strongly negatively contemporaneously correlated, the autoregressive OLS bias feeds into the small-sample distribution of $\hat{\beta}$ causing over-rejections of the null hypothesis of no predictability, $H_0: \beta = 0$. Moreover, notice even with the use of a possible finite-sample-bias correction on the OLS estimate, this reduces the noncentrality of the limiting distribution of the OLS t-statistic, but the distribution remains nonstandard in the near-integrated case. 

\begin{remark}
Assumption \ref{Assumption1} imposes a condition of a \textit{martingale difference sequence} (regressor predeterminedness) for the innovation term $u_t$ and allows to impose further assumptions regarding the modelling of conditional heteroscedasticity via volatility processes such as ARCH/GARCH (see, \cite{phillipsmagdal2009econometric}).  A Garch process is a typical example of a strong mixing process with exponential tails (see, \cite{meitz2023subgeometrically}). 
\end{remark}
Under Assumption \ref{Assumption1} a functional central limit theorem holds (\cite{phillips1992asymptotics}) 
\begin{align}
\frac{ \boldsymbol{x}_{  \floor{nr}  }   }{\sqrt{n} } \Rightarrow \boldsymbol{J}_{c}(r) := \int_{0}^r e^{ (r-s) \boldsymbol{C} } d \boldsymbol{W}(s),
\end{align}
which satisfies the following stochastic differential equation 
\begin{align*}
d \boldsymbol{J}_{c}(r) = c \boldsymbol{J}_{c}(r) dr + d \boldsymbol{W} (r), \ \ \ r \in (0,1). 
\end{align*}
Relevant notation can be found in several studies in the literature such as \cite{dickey1981likelihood}, \cite{jansson2006optimal}, \cite{mikusheva2012one} and \cite{demetrescu2020residual} among others. Furthermore, to avoid confusion regarding the properties of estimators based on a conditional mean functional form vis-a-vis a conditional quantile functional form, we include the relevant assumptions in this section. A similar approach of comparing estimation approaches across these two different modelling methods is presented by \cite{liu2023robust},  although in our framework this is might be practical given the gap in the literature on the relevant econometric issues we aim to tackle. The identification of the risk measures is based on definitions \eqref{A}-\eqref{B}.

\newpage

\begin{assumption}
\label{Assumption2}
The distribution function of innovations $u_t$, $F ( u )$ admits a continuous density function $f (u)$ away from zero on the domain $\mathcal{U} := \left\{  u: 0 < F(u) < 1 \right\}$. 
\end{assumption}

\begin{assumption}
\label{Assumption3}
Let $\big\{ \mathbb{P} \left(  Y_t < y | \boldsymbol{X}_t = x \right),  y \in \mathbb{R}, x \in \mathbb{R}^p \big\}$ 
be the conditional distribution function, denoted by $F (\cdot)$ which is absolutely continuous with respect to Lebesgue measure on $\mathbb{R}$ with continuous density functions $f_i ( \cdot ) \equiv f( \cdot | \boldsymbol{x}_i )$, where for each $i = 1,2,...$ $f_i (\cdot)$ is uniformly bounded away from zero and infinity.  
\end{assumption}

Then, the conditional quantile function of $y_t$ denoted with $\mathsf{Q}_{y_t} \left( \uptau | \mathcal{F}_{t-1} \right)$, replaces the conditional mean function of the predictive regression which implies the following model specification
\begin{align}
\mathsf{Q}_{y_t} \left( \uptau | \mathcal{F}_{t-1} \right) :=  F^{-1}_{ y_t | \boldsymbol{x}_{t-1} } (\uptau) \equiv \alpha (\uptau) + \boldsymbol{\beta} (\uptau)^{\prime} \boldsymbol{x}_{t-1}. 
\end{align}
such that $F_{ y_t | \boldsymbol{x}_{t-1} } (\uptau) := \mathbb{P} \big( y_t \leq \mathsf{Q}_{y_t} \left( \uptau | \mathcal{F}_{t-1} \right) \big| \mathcal{F}_{t-1} \big) \equiv \uptau$, where $\uptau \in (0,1)$ is some quantile level in the compact set $(0,1)$. Therefore, in order to define the innovation structure that corresponds to the quantile predictive regression, we employ the piecewise derivative of the loss function such that $\psi_{\uptau} ( \mathsf{u} ) = \big[ \uptau - \boldsymbol{1} \left\{ \mathsf{u} < 0 \right\} \big]$. Consequently, this implies that $u_t (\uptau) := u_{t} - F^{-1}_{ u } (\uptau)$ where $F^{-1}_{ u } (\tau)$  denotes the unconditional $\uptau-$quantile of the error term $u_{t}$. Then, the corresponding invariance principle for the nonstationary quantile predictive regression model is formulated as below  
\begin{align}
\frac{1}{\sqrt{n}}  \sum_{t=1}^{ \floor{ nr } } 
\begin{bmatrix}
\psi_{\uptau} \big( u_t (\uptau) \big) 
\\
\boldsymbol{v}_{t}
\end{bmatrix}
\Rightarrow
\begin{pmatrix}
B_{ \psi_{\uptau} } ( r )_{ ( 1 \times n )  } 
\\
\boldsymbol{B}_v ( r )_{ ( p \times n ) }
\end{pmatrix}
\equiv
\mathcal{BM} 
\begin{bmatrix}
\uptau (1 - \uptau) & \boldsymbol{\sigma}_{\psi_{\uptau} v}^{\prime}
\\
\boldsymbol{\sigma}_{v \psi_{\uptau}} & \boldsymbol{\Omega}_{vv}  
\end{bmatrix}
\end{align}   
\begin{assumption}
\label{Assumption4} 
Under Assumption \ref{Assumption2} and \ref{Assumption3}, then the following conditions for the innovation sequence hold (see, \cite{katsouris2023structural}):
\begin{itemize}
\item[\textbf{(\textit{i})}] The sequence of stationary conditional \textit{probability distribution functions (pdf)} denoted with $\big\{ f_{ u_t (\uptau), t-1}(.) \big\}$ evaluated at zero with a non-degenerate mean function such that $f_{ u_t (\uptau)  }(0) := \mathbb{E} \left[  f_{ u_t (\uptau), t-1}(0) \right] > 0$ satisfies a $\textit{FCLT}$ given as below
\begin{align}
\frac{1}{ \sqrt{n} } \sum_{t=1}^{ \floor{nr} } \big( f_{  u_t (\uptau), t-1}(0) - \mathbb{E} \left[  f_{ u_t (\uptau), t-1}(0) \right] \big) \Rightarrow B_{ f_{  u_t (\uptau) } } (r).
\end{align}
\item[ \textbf{(\textit{ii})} ]  For each $t$ and $\uptau \in (0,1)$, $f_{ u_t (\uptau), t-1}(.)$ is uniformly bounded away from zero with a corresponding conditional distribution function $F_t(.)$ which is absolutely  continuous with respect to Lebesgue measure on $\mathbb{R}$ (see, \cite{goh2009nonstandard} and \cite{lee2016predictive}).
\end{itemize}
\end{assumption}
where the following probability bound holds uniformly (see for example \cite{kulperger2005high}) 
\begin{align*}
\underset{ 0 \leq r \leq 1  }{ \mathsf{sup}   } \left|   \frac{1}{ n^{ 1 - \delta } } \sum_{t=1}^{ \floor{nr} } \left[  f_{ u_{t \uptau}, t -  1  } (0) - f_{ u_{\uptau}  } (0) \right] \right| = o_p(1). 
\end{align*}

\newpage 

\subsection{IVX filtration}

A solution to the \textit{endogeneity problem} discussed above, which provides robust statistical inference, is the IVX filtration proposed by \cite{phillipsmagdal2009econometric}. The IVX instrumentation  ensures instrument relevance while controlling for the abstract degree of persistence  (see, \cite{kostakis2015Robust}) and is constructed as below  
\begin{align}
z_t := \sum_{j=0}^{t-1} \rho_z^j \Delta x_{t-j}, \ \ \rho_z = \left( 1 - \frac{c_z}{n^{ \gamma_z} } \right), c_z >0 \ \text{and} \ \gamma_z \in (0,1). 
\end{align}  
The IVX filtration for $x_t$, is chosen so that $z_t$ corresponds to a mildly integrated process when the predictor $x_t$ is nearly integrated. Therefore, the IVX estimator of $\beta$ is found to have a slower convergence rate than the conventional OLS estimator under near integration, such that $n^{ \frac{1 + \gamma_z}{2} }$ (see, also  \cite{katsouris2023quantile} and  \cite{duffy2021estimation}). The asymptotic theory of the IVX estimator can be found to be mixed Gaussian irrespectively of the degree of endogeneity implied by $\delta$, leading to standard inference in $t $ and Wald tests. Moreover, it can be proved that under low persistence, the IVX estimator is asymptotically equivalent to the OLS procedure. Our main focus in this paper is the estimation of the risk measures of $\mathsf{VaR}$ and $\mathsf{CoVaR}$ using nearly integrated regressors, which is a novel aspect in the literature as conventional estimation approaches mainly consider that macroeconomic variables are stationary. Specifically, we start by considering the estimation of the $\mathsf{VaR}$ with a nearly integrated regressor. A vast literature has previously proposed robust methodologies for estimating the parameter of interest using the OLS approach and found evidence of significant size distortions when the model includes high persistence regressors.

\begin{example}
To obtain the IVX estimator for the linear predictive regression we employ the demeaned variates denoted with $y^{ \mu }_t = \boldsymbol{\beta}^{\prime} \boldsymbol{x}^{ \mu }_{t-1} + v^{ \mu }_t$, where the sample moments are estimated as $\bar{y}_t = \left( y_t - \frac{1}{n-1} \sum_{t=1}^n y_t \right)$ (e.g., see \cite{demetrescu2020residual} ). The IVX estimator\footnote{An alternative estimation approach to IVX which also ensures unified inference in the case of quantile-dependent parameters is proposed by \cite{liu2023unified}.  An augmented quantile predictive regression approach is currently work in progress by \cite{katsouris2023unified}} is 
\begin{align}
\widehat{ \boldsymbol{\beta} }^{ivx} = \left( \sum_{t=1}^n \bar{\boldsymbol{x}}_{t-1} \boldsymbol{z}_{t-1}^{\prime} \right)^{-1} \left( \sum_{t=1}^n  y_t \boldsymbol{z}_{t-1}^{\prime} \right).  
\end{align}
Denote with $\widehat{\boldsymbol{\beta}}^{ols}$ the OLS estimator and define the OLS residuals $\widehat{u}^{ols}_t = \bar{y}_t - \widehat{\boldsymbol{\beta}}^{ols} \bar{\boldsymbol{x}}_{t-1}$ and the residual variance estimator $\widehat{\sigma}_u^2 = \frac{1}{n-1} \sum_{t=1}^n \widehat{u}_{t}^2$. Moreover, the IVX-Wald statistic under the null hypothesis is $\mathcal{W} = \widehat{\boldsymbol{\beta}}^{ivx \prime} \widehat{\boldsymbol{V}}^{-1} \widehat{\boldsymbol{\beta}}^{ivx}$, where a feasible estimator for the covariance matrix is given by
\begin{align}
\widehat{\boldsymbol{V}}^{ivx} 
= 
\left( \sum_{t=1}^n \bar{\boldsymbol{x}}_{t-1} \boldsymbol{z}_{t-1}^{\prime} \right)^{-1} \left( \widehat{\sigma}_u^2 \sum_{t=1}^n \boldsymbol{z}_{t-1} \boldsymbol{z}_{t-1}^{\prime} \right) \left( \sum_{t=1}^n  \bar{\boldsymbol{x}}_{t-1} \boldsymbol{z}_{t-1}^{\prime} \right)^{-1}. 
\end{align}
if our main interest is testing the null hypothesis of no predictability.
\end{example}

\newpage 

\subsection{Two-Stage Estimation Procedure}

\begin{assumption}
Suppose that Assumptions 1 to 3 hold, then based on the conditional quantile functional form of the nonstationary quantile predictive regression, all quantile-dependent parameters are assumed to be constant for any fixed quantile level $\tau \in \mathcal{T} = (0,1)$.  
\end{assumption}

\subsubsection{First-Stage Procedure}

During the first stage of our procedure we obtain a consistent estimator of the parameter vector. We focus on the IVX estimator which is found to be robust to the abstract degree of persistence (nearly integrated regressors). The econometric model is the linear predictive regression as below 
\begin{align}
y_{1,t} &= \beta_{01} (\uptau)  + \boldsymbol{\beta}_{11}^{\prime} (\uptau) \boldsymbol{x}_{1,t-1} + u_{1t}(\uptau), \ \ \ \text{for} \ \ \ t = 1,...,n
\\
\boldsymbol{x}_{1,t} &= \boldsymbol{R}_{1,n} \boldsymbol{x}_{1,t-1} + \boldsymbol{v}_{1,t}
\end{align}
Denote with $\boldsymbol{\beta} (\uptau) = \left(\beta_{01}(\uptau), \boldsymbol{\beta}_{11}^{\prime} (\uptau) \right)^{\prime}$ and the corresponding IVX estimator with $\boldsymbol{\beta}^{ivx} (\uptau)$. The asymptotic behaviour of the IVX estimator that corresponds to the linear predictive regression model under abstract degree of regressors persistence is studied by \cite{kostakis2015Robust} while the asymptotic properties of the IVX estimator for the quantile predictive regression model is studied by \cite{lee2016predictive}.  Since when estimating the risk measure pair $( \mathsf{VaR}, \mathsf{CoVaR})$, we use the conditional quantile distribution, i.e., to capture the effect of $\boldsymbol{x}_{t-1}$ on the conditional quantile of $\boldsymbol{y}_{t}$; however with the additional assumption of possibly nonstationary regressors  (see, also \cite{koenker2004unit} and \cite{xiao2009quantile}), then estimation relies on the nonstationary quantile predictive regression.

\subsubsection{Second-Stage Procedure} During the second stage of our procedure we consider a consistent estimator of the parameter vector for the set of regressors which includes the nonstationary regressors as well as the generated regressor from the first stage procedure. More precisely, the generated regressor in our study corresponds to the fitted values of the predictive regression model based on the IVX estimator obtained in the first stage procedure. 
\begin{align}
y_{2,t} &= \beta_{02}(\uptau) + \boldsymbol{\beta}_{12}^{\prime} (\uptau) \boldsymbol{x}_{2,t-1} + \delta (\uptau)\widehat{y}_{1,t} \big( \widehat{\boldsymbol{\beta}}^{ivx} (\uptau)  \big) + u_{2t}(\uptau), \ \ \ \text{for} \ \ \ t = 1,...,n 
\\
\boldsymbol{x}_{2,t} &= \boldsymbol{R}_{2,n} \boldsymbol{x}_{2,t-1} + \boldsymbol{v}_{2,t}, \ \ \boldsymbol{R}_{jn} = \left(  \boldsymbol{I}_p - \frac{\boldsymbol{C}_{jp}}{n^{\gamma}} \right), \ \gamma = 1.
\end{align}
where $\boldsymbol{C}_{jp} =  \mathsf{diag} \left\{ c_{j1},..., c_{jp} \right\}$ and $j \in \left\{ 1, 2 \right\}$. Define with $\widetilde{\boldsymbol{\beta}} (\uptau) = \big( \beta_{02}(\uptau), \boldsymbol{\beta}_{12}^{\prime}(\uptau), \delta (\uptau) \big)^{\prime}$ the parameter vector of the second stage procedure and the corresponding IVX estimator with $\widetilde{\boldsymbol{\beta}}^{ivx}(\uptau)$. 

Specifically, we observe that this extended parameter vector includes both the nearly integrated regressors as well as the generated regressor from the first stage estimation step.

\newpage

In this case, we need to develop the asymptotic distribution theory for both the OLS and IVX estimators when the generated regressor is included in the set of regressors in the second stage predictive regression model that corresponds to the econometric specification of the CoVaR risk measure. A modified IVX estimator is necessary to be developed in order to account for the presence of the particular effect in the setting of the quantile predictive regression models. The nuisance coefficient of persistence is defined such that $c_i > 0$ and is either $\gamma = 1$, which corresponds to near unit root regressors or $\gamma \in (0,1)$ that corresponds to mildly integrated regressors. 

Therefore, the focus of this paper is the implementation of the above econometric environment and the development of the corresponding asymptotic theory in the case of the quantile predictive regression model. Consider the piecewise derivative of the loss function which is defined as below $
\psi_{\uptau} \left( u \right) = \uptau - \mathbf{1} \left( u < 0 \right)$. Then the innovation sequence of the quantile predictive regression model, $u_t ( \uptau ) = u_t - F_{u}^{-1} ( \uptau )$, and $F_{u}^{-1} ( \uptau )$ is the unconditional $\uptau$-quantile of $u_t$, where $\uptau \in (0,1)$ is a fixed quantile (see, \cite{lee2016predictive}, \cite{fan2019predictive} and \cite{katsouris2022asymptotic}). Thus, the $\boldsymbol{\beta}^{ivx}$ estimator is obtained from the first stage procedure using the IVX filter\footnote{The reason that we apply the IVX instrumentation procedure only to the estimate from the first stage quantile predictive regression model rather to the corresponding generated covariate is to ensure that those fitted values can preserve their definition as the estimated Value-at-Risk but in our case adjusted based on the presence of persistent predictors and corrected accordingly using the IVX methodology.} for the model coefficients based on the nonstationary regressors which implies that the generated regressor is defined as
\begin{align}
\widehat{y}_{1,t} \left( \widehat{\boldsymbol{\beta}}^{ivx} (\uptau) \right)
=
\boldsymbol{x}_{1,t-1}^{\prime} \widehat{\boldsymbol{\beta}}^{ivx}(\uptau)
\end{align}
Then, to conduct inference we need to obtain a consistent estimator for the covariance matrix of $\tilde{ \boldsymbol{\beta} }_2^{ivx}$ using the usual "sandwich" formula given below (see, \cite{demetrescu2020residual}).

\section{Doubly Corrected Estimation Method: Conditional Mean Case}
\label{Section3}

\subsection{Residual Augmented IVX Estimation Example}

We begin our asymptotic theory analysis by considering the framework for the residual augmented IVX estimation methodology previously examined in the literature (see, \cite{demetrescu2020residual}). The IVX instrumentation provides a methodology such that the asymptotic distribution of the IVX-Wald based test statistic converges to a nuisance-parameter free distribution. Consider the feasible predictive regression defined with the following form 
\begin{align}
y^{\mu}_{t} = \beta x_{t-1}^{\mu} + \delta \hat{v}_t + \varepsilon_t, \ \ \  \varepsilon_t \sim_{\textit{i.i.d}} \ \text{and} \ \ t = 1,...,n. 
\end{align}
where $y^{\mu}_{t}$ and $x_{t-1}^{\mu}$ denote the demeaned variates, following the notation employed in several studies in the literature.  Moreover, $\hat{v}_t$ represents a generated covariate which in our study corresponds to the estimated $\mathsf{VaR}$ under nonstationarity.

\newpage

\begin{remark}
The parameter of interest for inference purposes is the slope coefficient of the quantile predictive regression with the augmented regressor. The IVX estimator of $\beta$, has a slower convergence rate than the conventional OLS estimator under near integration, but for low persistence it is asymptotically equivalent to the OLS estimator. However, the IVX estimator is mixed Gaussian in the limit regardless of the degree of endogeneity which implies standard inference based on the Wald test. Furthermore, we conjecture that the generated regressor since it is stationary it will not affect the convergence of the IVX estimator to its asymptotic limit, given the lower convergence rate of IVX compared to the OLS estimator. Roughly speaking our two-stage estimation procedure has some similarities with the residual-based augmentation of \cite{demetrescu2020residual}. Nevertheless the objective of our asymptotic theory analysis is to demonstrate that in finite and large samples such good properties hold under the presence of the generated covariate. 
\end{remark}

\subsection{Generated Regressor in Nonstationary Linear Predictive Regression}

Recall that the above expressions correspond to the model estimates of  linear predictive regression model with conditional mean functional form, while the risk measures require to use the quantile predictive regression framework. In any case, to keep track of our derivations we begin our analysis using the linear predictive regression model and the framework of residual-based augmentation of \cite{demetrescu2020residual}, but we replace the regressor $\hat{v}_t$ with the generated dependent variable, $\widehat{ \bar{y} }_{1,t}$, which is used as a proxy for $\mathsf{VaR}$ under time series nonstationarity

In compact form we have that 
\begin{align}
\bar{y}_{2,t} = \boldsymbol{X}_{t-1} \boldsymbol{\Gamma}^{\prime} + \varepsilon_t, \ \ \  \varepsilon_t \sim_{\textit{i.i.d}} \ \text{and} \ \ t = 1,...,n. 
\end{align}
where $\boldsymbol{X}_{t-1} = \big( \bar{\boldsymbol{x}}_{2t-1} \  \widehat{ \bar{y} }_{1,t} \big)$ and $\boldsymbol{\Gamma} = \big( \boldsymbol{\beta} \ \delta \big)$. Moreover, define with $\boldsymbol{Z}_{t-1} = \big( \boldsymbol{z}_{2t-1} \ \widehat{ \bar{y} }_{1,t}  \big)$ where $\boldsymbol{z}_{t-1}$ is the IVX instrument for $\boldsymbol{x}_{2t-1}$  and the generated regressor is self-instrumented.

The generated covariate $\widehat{ \bar{y} }_{1,t}$ corresponds to the estimated $\mathsf{VaR}$ under nonstationarity.  In practise, these nonstationary predictors can be different than the nonstationary predictors corresponding to the regressand $\widehat{y}_{2,t}$ due to the proposed dependence structure. However, for simplicity we can also consider the case in which we use the same nonstationary regressors for both predictive regression models. In this section, we obtain derivations for the linear predictive regression model. However, when estimating the risk measures of $\mathsf{VaR}$ and $\mathsf{CoVaR}$, the nonstationary quantile predictive regression models are used which implies that all parameters are quantile-dependent for a fixed quantile $\uptau \in (0,1)$. Therefore, we are interested for the consistent estimation of the $\beta_2$ estimator (e.g., via IVX) as well as constructing hypothesis testing for: (i) only $\beta$ parameter (e.g., no predictability), (ii) only $\delta$ parameter (e.g., no presence of generated regressor) and (iii) linear restrictions on all the parameters of the predictive regression model that includes both nonstationary regressors and the generated regressor (e.g., univariate generated regressor).

\newpage
The covariance estimator for $\tilde{ \boldsymbol{\beta} }_2^{ivx}$ is expressed as below
\begin{align*}
\widehat{ \mathsf{Cov} \left(  \tilde{ \boldsymbol{\beta} }_2^{ivx} \right) } 
&:=
\left( \boldsymbol{B}_n^{-1} \right) \mathbb{V}_n  \left( \boldsymbol{B}_n^{-1} \right)^{\prime}, \ \ \
\boldsymbol{B}_n := \left( \sum_{t=1}^n \bar{\boldsymbol{x}}_{t-1} \boldsymbol{z}_{t-1}^{\prime} \right)
\\
\mathbb{M}_n 
&:= 
\left( \sum_{t=1}^n \boldsymbol{z}_{t-1} \boldsymbol{z}_{t-1}^{\prime} \tilde{\varepsilon}^2_t \right) 
+  
\left[ \boldsymbol{\gamma}^{\top} \otimes \left( \frac{1}{n} \sum_{t=1}^n \boldsymbol{z}_{t-1} \bar{\boldsymbol{x}}_{t-1}^{\top} \right) \left( \frac{1}{n} \sum_{t=1}^n \bar{\boldsymbol{x}}_{t-1} \bar{\boldsymbol{x}}_{t-1}^{\top} \right)^{-1} \right]
\left( \sum_{t=1}^n \boldsymbol{\nu}_t \boldsymbol{\nu}_t^{\top} \otimes \bar{\boldsymbol{x}}_{t-1} \bar{\boldsymbol{x}}_{t-1}^{\top} \right)
\\
&\times
\left[ \boldsymbol{\gamma} \otimes \left( \frac{1}{n} \sum_{t=1}^n \bar{\boldsymbol{x}}_{t-1} \bar{\boldsymbol{x}}_{t-1}^{\top} \right)^{-1} \left( \frac{1}{n} \sum_{t=1}^n \bar{\boldsymbol{x}}_{t-1} \boldsymbol{z}_{t-1}^{\top} \right) \right].
\end{align*}
Therefore, the just-identified model for the combined IVX estimator of $\boldsymbol{\Gamma}$ is given by 
\begin{align}
\hat{\boldsymbol{\Gamma}} = \left( \sum_{t=1}^n \boldsymbol{Z}_{t-1}^{\prime} \boldsymbol{X}_{t-1} \right)^{-1} \left(  \sum_{t=1}^n \boldsymbol{Z}_{t-1}^{\prime} \bar{y}_{2,t} \right).
\end{align}  
Furthermore, notice that the nonstationary regressors $x_{t-1}$ and the generated regressor are orthogonal since the generated regressor can be considered as exogenous variable to the particular model, then we can compute the generated regressor augmented estimator and test statistic in two steps. Notice the constructed test statistic refers to the IVX estimator. Therefore, we focus on deriving the asymptotic distribution for the IVX estimator under the presence of the generated regressor in the quantile predictive regression model. The estimation procedure is implemented into two-stages, and our asymptotic theory analysis demonstrates that although the limit converges into a nonstandard distribution the particular limit is nuisance-parameter free. 
\begin{remark}
The key point in this estimation step, is that the variance-covariance needs to be adjusted to account for the first-step estimation error. In other words, when estimating the VaR using the quantile predictive regression of the first stage, the forecast at the current time period $t$, is estimated with some additional source of error. For this reason, an adjustement is necessary. Consider the estimated coefficient of the quantile predictive regression model of the first stage as
\begin{align}
Q_{\uptau} \big( y_{2t} | x_{t-1}, y_{1t} \big) = \theta X_{t-1}, \ \ \ \uptau \in (0,1).     
\end{align}
\end{remark}
The estimation step needs to account for two sources of errors, that is, the usaul estimation error in obtaining a consistent estimator that corresponds to the IVX instrumentation of the nearly integrated regressors, and the second source of error is the sampling error in generating the forecast for the VaR from the first stage of the process. More specifically, this implies that using a Bahadur representation of the QR-IVX estimator we need to determine the precise stochastic order of the remainder term when the generated regressor is included in the conditional quantile specification of the model. Although when the parameter vector from the first stage is obtained via the IVX filtration, it has the usual convergence rate that the IVX estimator has, an estimation error carries in the second-stage estimation procedure, which requires us to consider a suitable correction to the overall variance due to the presence of the generated regressor.

\newpage

Moreover, we expect that the stochastic equicontinuity property to still hold regardless of the plug-in estimation approach and the presence of time series nonstationarity.

\subsubsection{Estimation Procedure in Nonstationary Linear Predictive Regression Case}

\begin{itemize}

\item[\textbf{Step 1.}] OLS-regress $\bar{y}_{2,t}$ on the generated regressor $\widehat{ \bar{y} }_{1,t} $ and obtain the OLS estimator below
\begin{align}
\widehat{ \delta }^{ols}
=
\left( \sum_{t=1}^n \widehat{ \bar{y} }^{2} _{1,t}  \right)^{-1} \left( \sum_{t=1}^n  \widehat{ \bar{y} }_{1,t} \bar{y}_{2,t} \right). 
\end{align} 
and compute the residual such that $\widehat{\eta}_t = \bar{y}_{2,t} - \widehat{ \delta }^{ols}\widehat{ \bar{y} }_{1,t}$. Here we assume that the elements of the dependent vector are not serially correlated, which implies that $y_{1t}$ and $y_{2t}$ are uncorrelated. This simplifies the development of the asymptotic theory, despite the presence of endogeneity and highly persistent regressors.

\item[\textbf{Step 2.}] IVX-regress $\bar{y}_{2,t}$ on $\bar{\boldsymbol{x}}_{t-1}$, leading to 
\begin{align}
\widehat{ \boldsymbol{\beta} }_2^{ivx} = \left( \sum_{t=1}^n \bar{ \boldsymbol{x} }_{t-1} \boldsymbol{z}_{t-1}^{\prime} \right)^{-1} \left(  \sum_{t=1}^n \tilde{y}_t  \boldsymbol{z}_{t-1}^{\prime} \right).
\end{align}

\end{itemize}
Recall the linear predictive regression model (for node 1) is given by $y_{1,t} = \boldsymbol{\beta}_1^{\prime} \boldsymbol{x}_{t-1} + u_{1,t}$ and $\boldsymbol{x}_t = \boldsymbol{R} \boldsymbol{x}_{t-1} + \boldsymbol{v}_{1,t}$. Then we estimate the following model 
\begin{align}
\bar{y}_{2,t} &=  \boldsymbol{\beta}_2^{\prime} \bar{\boldsymbol{x}}_{t-1} + \delta \widehat{ \bar{y} }_{1,t} + u_{2,t}
\end{align}
Moreover, by OLS-regressing $\bar{y}_{2,t}$ on $\widehat{ \bar{y} }_{1,t}$ such that $\bar{y}_{2,t} = \theta \widehat{ \bar{y} }_{1,t} + \eta_t$, where $\eta_t$ is a disturbance term, we construct the transformed dependent variable $\tilde{y}_{2,t}$ given by the following expression
\begin{align}
\tilde{y}_{2,t} := \big(  \bar{y}_{2,t} - \widehat{\theta}^{ols} \widehat{ \bar{y} }_{1,t} \big)
\end{align}

\begin{itemize}
\item In the second stage of the estimation procedure, we need to re-estimate the quantile kernel function when the generated regressor is included in the design matrix, especially when constructing the Wald test under the null hypothesis of no systemic risk in the network. To obtain a measure of the performance of the testing procedure, we need to compare the standard estimation procedure without the variance-covariance matrix correction against the estimation procedure that includes the adjustment, in both cases under the null hypothesis. Then, under the alternative hypothesis of deviations, that is, existence of non-zero systemic risk effect in the network (see, \cite{katsouris2021optimal, katsouris2023statistical}), then the larger this distance is, under the adjustement then the larger the power function of the test would be when all the individual quantile predictive regression models of the system correspond to a non-zero coefficient of node-specific systemic risk. 

\item When the nonstationary regressors for the two nodes are the same then second-order effects can contribute to the covariance of the estimator (estimation bias). However, due to the proposed dependence structure (see, \cite{katsouris2021optimal, katsouris2023statistical}), these possibly nonstationary regressors are assumed to be different across the system-specific equations. Moreover, the generated covariate, denoted by, $\widehat{y}_{1,t}$, corresponds to the risk measure estimate of $\mathsf{VaR}$, which can have a direct impact on the convergence rate for consistent estimation of the model parameter $\delta$, during the second-stage estimation procedure. In summary, we need to control for two effects, that is, a bias effect on the estimate of $\delta$ and a variance effect. The bias effect will disappear when $\boldsymbol{x}_{t-1}$ and the regressors used to estimate $\widehat{y}_{1,t}$ are uncorrelated. The variance effect appears due to the choice of the generated regressors. Standard OLS estimation methods do not correct for the particular variance effect. However, the variance effect does not influence the inference methodology for the estimator $\hat{\beta}_2$, in which case our asymptotic theory analysis shall investigate that a consistent estimator is obtained regardless of the presence of nonstationarity. Thus, we use the generated covariate, $\widehat{y}_{1,t}$, from the first-stage estimation procedure as an additional regressor to the nonstationary regressors that correspond to node 2, for estimating the risk measure of $\mathsf{CoVaR}$. 
\end{itemize}
Therefore, the IVX estimator during the second-stage estimation procedure ($\mathsf{CoVaR}$)  is obtained by the following expression: 
\begin{align*}
\tilde{ \boldsymbol{\beta} }_2^{ivx} 
&= 
\left( \sum_{t=1}^n  \bar{\boldsymbol{x}}_{2,t-1}\boldsymbol{z}_{2,t-1}^{\prime} \right)^{-1} \left(  \sum_{t=1}^n  \tilde{y}_{2,t} \boldsymbol{z}_{2,t-1}^{\prime} \right)
\nonumber
\\
&=  
\left( \sum_{t=1}^n  \bar{\boldsymbol{x}}_{2,t-1}\boldsymbol{z}_{2,t-1}^{\prime} \right)^{-1} \left(  \sum_{t=1}^n  \bigg[ \boldsymbol{\beta}_2^{\prime} \bar{\boldsymbol{x}}_{2,t-1} + \delta \widehat{y}_{1,t} + u_{2,t} - \widehat{\delta}^{ols} \widehat{y}_{1,t} \bigg] \boldsymbol{z}_{2,t-1}^{\prime} \right)
\nonumber
\\
&=
\boldsymbol{\beta}_2 
- 
\left( \widehat{\delta}^{ols} -  \delta \right) \left( \sum_{t=1}^n \bar{\boldsymbol{x}}_{2,t-1}\boldsymbol{z}_{2,t-1}^{\prime}  \right)^{-1} \left(  \sum_{t=1}^n  \widehat{y}_{1,t} \boldsymbol{z}_{2,t-1}^{\prime}  \right)
+ 
\left( \sum_{t=1}^n \bar{\boldsymbol{x}}_{2,t-1}\boldsymbol{z}_{2,t-1}^{\prime}  \right)^{-1} \left( \sum_{t=1}^n u_{2,t} \boldsymbol{z}_{2,t-1}^{\prime}  \right).
\end{align*}
By rearranging we obtain that 
\begin{align*}
\bigg( \tilde{ \boldsymbol{\beta}}_2^{ivx} -  \boldsymbol{\beta}_2 \bigg)
= 
- \left( \widehat{\delta^{ols}} - \delta \right) \left( \sum_{t=1}^n \bar{\boldsymbol{x}}_{2,t-1}\boldsymbol{z}_{2,t-1}^{\prime} \right)^{-1} \left(  \sum_{t=1}^n \widehat{y}_{1,t} \boldsymbol{z}_{2,t-1}^{\prime}  \right)
+ 
\left( \sum_{t=1}^n \bar{\boldsymbol{x}}_{2,t-1}\boldsymbol{z}_{2,t-1}^{\prime} \right)^{-1} \left( \sum_{t=1}^n  u_{2,t} \boldsymbol{z}_{2,t-1}^{\prime} \right).
\end{align*}
The first term of the above expression is simplified as below 
\begin{align}
\left( \sum_{t=1}^n \bar{\boldsymbol{x}}_{2,t-1}\boldsymbol{z}_{2,t-1}^{\prime}  \right)^{-1} \left(  \sum_{t=1}^n \widehat{ \boldsymbol{\beta}}_1^{ivx} \bar{\boldsymbol{x}}_{1,t-1}  \boldsymbol{z}_{2,t-1}^{\prime} \right) 
\equiv 
\widehat{ \boldsymbol{\beta}}_1^{ivx}.
\end{align}
which holds when we have identical nonstationary regressors such that $\bar{\boldsymbol{x}}_{1,t-1} \equiv \bar{\boldsymbol{x}}_{2,t-1}$. Then,  
\begin{align}
\bigg( \tilde{ \boldsymbol{\beta}}_2^{ivx} -  \boldsymbol{\beta}_2 \bigg)
&= 
- \left( \widehat{\delta}^{ols} - \delta \right)  \times \widehat{ \boldsymbol{\beta}}_1^{ivx} 
+ 
\left( \sum_{t=1}^n \bar{\boldsymbol{x}}_{2,t-1}\boldsymbol{z}_{2,t-1}^{\prime}  \right)^{-1} \left( \sum_{t=1}^n u_{2,t} \boldsymbol{z}_{2,t-1}^{\prime}   \right).
\end{align}

\newpage

\begin{remark}
Notice that when $\boldsymbol{x}_{t-1}$ and $\widehat{y}_{1,t} $ are not correlated, then $\widehat{\theta}^{ols}$ is a consistent estimator of $\delta$. However, in the case we employ the same nonstationary predictors to estimate the value of $\widehat{y}_{1,t}$ then these two quantities are correlated which requires to consider the stochastic dominance of terms due to the presence of different effects. Nevertheless, the network dependence structure proposed by \cite{katsouris2021optimal, katsouris2023statistical} ensures identification and consistent estimation of model coefficients which are important for inference purposes. 
\end{remark}

Recall that $ \widehat{ \boldsymbol{\beta}}_1^{ivx}$ corresponds to the IVX estimator of the quantile predictive regression model for the $\mathsf{VaR}$ risk measure. From Lemma B4 of \cite{kostakis2015Robust} it holds that
\begin{align}
\frac{1}{ n^{ \frac{1 + \gamma_z }{2} } } \sum_{t=1}^n \boldsymbol{z}_{2,t-1} u_{2,t} \Rightarrow \mathcal{N} \left( 0, \sigma_{ u_2 }^2 \times \boldsymbol{V}_{c_z} \right).
\end{align} 
Therefore, by incorporating the convergence rate we have that 
\begin{align*}
\label{the.exp}
n^{ \frac{ 1 + \gamma_z }{2} } \bigg( \tilde{ \boldsymbol{\beta} }_2^{ivx} - \boldsymbol{\beta}_2 \bigg)
= 
n^{ \frac{ 1 + \gamma_z }{2} }   \left( \delta - \widehat{\delta}^{ols} \right)  \times  \widehat{ \boldsymbol{\beta}}_1^{ivx}
+ 
\left(  \frac{1}{ n^{ 1 + \gamma_z } }    \sum_{t=1}^n \bar{\boldsymbol{x}}_{2,t-1}\boldsymbol{z}_{2,t-1}^{\prime} \right)^{-1} \left( \frac{1}{ n^{ \frac{ 1 + \gamma_z }{2} } } \sum_{t=1}^n \boldsymbol{z}_{2,t-1} u_{2,t} \right).
\end{align*}
Thus, for the LUR regressor case we have that the second term of expression \eqref{the.exp}  converges to
\begin{align*}
&\left(  \frac{1}{ n^{ 1 + \gamma_z } }   \sum_{t=1}^n \bar{\boldsymbol{x}}_{2,t-1}\boldsymbol{z}_{2,t-1}^{\prime} \right)^{-1} \left( \frac{1}{ n^{ \frac{ 1 + \gamma_z }{2} } } \sum_{t=1}^n \boldsymbol{z}_{2,t-1} u_{2,t} \right)
\Rightarrow
-
\left( \boldsymbol{\Omega}_{uu} + \int_0^1 \boldsymbol{J}^{\mu}_c  d J^{\prime}_c \right) \boldsymbol{C}_z^{-1}  \times  \mathcal{N} \left( 0, \sigma_{ u_2 }^2 \times \boldsymbol{V}_{C_z} \right).
\end{align*}
To see this, from the Appendix of KMS expression (27) gives that 
\begin{align}
\left(  \frac{1}{ n^{ 1 + \gamma_z } } \sum_{t=1}^n \bar{\boldsymbol{x}}_{2,t-1}\boldsymbol{z}_{2,t-1}^{\prime} \right) 
=  
\bigg( \boldsymbol{\Omega}_{uu} + \boldsymbol{V}_C \boldsymbol{C} \bigg) \boldsymbol{C}_z^{-1}  + o_p(1).
\end{align}
where the stochastic covariance matrices are defined as below
\begin{align}
\boldsymbol{V}_{C} = \int_0^{\infty} e^{r \boldsymbol{C} } \boldsymbol{\Omega}_{uu} e^{r \boldsymbol{C} } dr \ \ \ \text{and} \ \ \  \boldsymbol{V}_{C_z} = \int_0^{\infty} e^{r \boldsymbol{C}_z } \boldsymbol{\Omega}_{uu} e^{r \boldsymbol{C}_z } dr.
\end{align}
Notice that by \cite{kostakis2015Robust} we have that $\boldsymbol{\Psi}_{uu} = \left( \boldsymbol{\Omega}_{uu} + \displaystyle \int_0^1 \boldsymbol{J}^{\mu}_c  d J^{\prime}_c \right)$ since we consider local unit root regressors. Therefore, similar to Theorem A (i) of KMS it holds that 
\begin{align}
\left(  \frac{1}{ n^{ 1 + \gamma_z } }  \sum_{t=1}^n \bar{\boldsymbol{x}}_{2,t-1}\boldsymbol{z}_{2,t-1}^{\prime} \right)^{-1} \left( \frac{1}{ n^{ \frac{ 1 + \gamma_z }{2} } } \sum_{t=1}^n \boldsymbol{z}_{2,t-1} u_{2,t} \right) 
\nonumber
\Rightarrow
\mathcal{MN} \bigg( 0, \sigma_{ u_2 }^2 \times \left( \boldsymbol{\Psi}_{u_2u_2}^{-1} \right)^{\prime} \boldsymbol{C}_z  \boldsymbol{V}_{C_z} \boldsymbol{C}_z \boldsymbol{\Psi}_{u_2 u_2 }^{-1} \bigg)
\end{align}
which is a mixed Gaussian random variate since the covariance matrix is a function of the OU process. Next we investigate the asymptotic behaviour of the first term of expression \eqref{the.exp}: 
\begin{align}
\mathcal{A} := 
 n^{ \frac{ 1 + \gamma_z }{2} }  \left( \delta - \widehat{\delta}^{ols} \right)  \times  \widehat{ \boldsymbol{\beta}}_1^{ivx}. 
\end{align}

\newpage

Specifically, assuming that node 1 and 2 have a set of nonstationary regressors that are generated from non-identical stochastic processes (i.e., not identical regressors), then it follows that
\begin{align}
\mathcal{A} 
&=  
\left( \sum_{t=1}^n \bar{\boldsymbol{x}}_{2,t-1}\boldsymbol{z}_{2,t-1}^{\prime}  \right)^{-1} \left(  \sum_{t=1}^n  \bar{\boldsymbol{x}}_{1,t-1}  \boldsymbol{z}_{2,t-1}^{\prime} \right) 
\times 
\widehat{  \boldsymbol{\beta}_1^{ivx} }
\nonumber
\\
&=
\left( \sum_{t=1}^n \bar{\boldsymbol{x}}_{2,t-1}\boldsymbol{z}_{2,t-1}^{\prime}  \right)^{-1} \left(  \sum_{t=1}^n  \bar{\boldsymbol{x}}_{1,t-1}  \boldsymbol{z}_{2,t-1}^{\prime} \right) 
\nonumber
\times
\left[ \boldsymbol{\beta}_1 
+ 
\left(  \sum_{t=1}^n \bar{\boldsymbol{x}}_{1,t-1} \boldsymbol{z}_{1,t-1}^{\prime} \right)^{-1} \left(  \sum_{t=1}^n \bar{u}_{1,t} \boldsymbol{z}_{1,t-1}^{\prime}  \right)  \right].
\end{align}
Regarding the asymptotic behaviour of the first term of the expression we have that the IVX estimator from the quantile predictive regression that corresponds to the VaR risk measure is 
\begin{align*}
\widehat{ \boldsymbol{\beta} }_1^{ivx} 
&=  
\left(  \sum_{t=1}^n \bar{\boldsymbol{x}}_{1,t-1} \boldsymbol{z}_{1,t-1}^{\prime} \right)^{-1} \left(  \sum_{t=1}^n y_{1,t} \boldsymbol{z}_{1,t-1}^{\prime} \right),  \ \ \ \bar{y}_{1,t} = \boldsymbol{\beta}_1^{\prime} \bar{\boldsymbol{x}}_{1,t-1} + \bar{u}_{1,t}
\\
&= 
\boldsymbol{\beta}_1 
+ 
\left(  \sum_{t=1}^n \bar{\boldsymbol{x}}_{1,t-1} \boldsymbol{z}_{1,t-1}^{\prime} \right)^{-1} \left(  \sum_{t=1}^n \bar{u}_{1,t} \boldsymbol{z}_{1,t-1}^{\prime}  \right). 
\end{align*}
Moreover, the number of nonstationary regressors for both econometric specifications is the same although these two sets of regressors possibly have different nuisance parameters of persistence. The limiting distributions for their IVX estimators have the same dimensions such that 
\begin{align*}
n^{ \frac{ 1 + \gamma_z }{2} } \left( \widehat{ \boldsymbol{\beta} }_1^{ivx} - \boldsymbol{\beta}_1 \right)
&= 
\left(  \frac{1}{ n^{ 1 + \gamma_z } }  \sum_{t=1}^n \bar{\boldsymbol{x}}_{1,t-1} \boldsymbol{z}_{1,t-1}^{\prime} \right)^{-1} \left(\frac{1}{ n^{ \frac{ 1 + \gamma_z }{2} } } \sum_{t=1}^n \bar{u}_{1,t} \boldsymbol{z}_{1,t-1}^{\prime}  \right) 
\\
&\Rightarrow
\mathcal{MN} \bigg( 0, \sigma_{ u_1 }^2 \times \left( \boldsymbol{\Psi}_{u_1u_1}^{-1} \right)^{\prime} \boldsymbol{C}_z  \boldsymbol{V}_{C_z} \boldsymbol{C}_z \boldsymbol{\Psi}_{u_1u_1}^{-1} \bigg).
\end{align*}
Rearranging the above expressions, we obtain the following limit result
\begin{align}
 n^{ \frac{ 1 + \gamma_z }{2} }  \bigg[ \bigg( \tilde{ \boldsymbol{\beta} }_2^{ivx} - \boldsymbol{\beta}_2 \bigg) - \boldsymbol{\beta}_1 \bigg] 
&= 
\left( \delta - \widehat{\delta}^{ols} \right)  \times \mathcal{MN} \bigg( 0, \sigma_{ u_1 }^2 \times \left( \boldsymbol{\Psi}_{u_1u_1}^{-1} \right)^{\prime} \boldsymbol{C}_z  \boldsymbol{V}_{C_z} \boldsymbol{C}_z \boldsymbol{\Psi}_{u_1u_1}^{-1} \bigg)
\nonumber
\\
&\ \ \ \ + 
\mathcal{MN} \bigg( 0, \sigma_{ u_2 }^2 \times \left( \boldsymbol{\Psi}_{u_2u_2}^{-1} \right)^{\prime} \boldsymbol{C}_z  \boldsymbol{V}_{C_z} \boldsymbol{C}_z \boldsymbol{\Psi}_{u_2 u_2 }^{-1} \bigg).
\end{align}
In this paper we are not proposing any novel solution to the endogeneity issue in quantile predictive regressions models. We employ a well-investigated method in the literature that tackles the endogeneity problem, meaning it produces weak convergence results into standard asymptotic distributions (see, \cite{lee2016predictive}, \cite{fan2019predictive}). Alternative approaches that tackle the endogeneity problem and produce uniform valid inference regardless of the unknown persistence properties are proposed in the frameworks of \cite{cai2023new} and \cite{liu2023unified}.  In contrast, our study provides a general framework and establish the asymptotic properties for quantile predictive regression models with a generated regressor, which is particularly useful when jointly estimating the risk measure pair of $( \mathsf{VaR}, \mathsf{CoVaR} )$ under the presence of time series nonstationarity. 

\newpage

In fact, regardless of the estimation method employed to robustify the quantile-based model parameters to the unknown persistence, when estimating the $\mathsf{CoVaR}$ as per \citep{AB2016covar} and \cite{hardle2016tenet}, under time series nonstationarity (LUR parametrization), similar econometric issues need to be tackled to ensure relevant statistical properties still hold.

\subsubsection{Asymptotic Properties of Stage-2 Estimator}

Next, consider for a moment the limiting distribution of $\left( \widehat{\delta}^{ols} - \delta \right)$ where $\bar{y}_{2,t} =  \delta  \widehat{ y }_{1,t}  + \eta_t$. Therefore, the OLS estimator of $ \delta$ is given by $\widehat{ \delta }^{ols} = \displaystyle \left( \sum_{t=1}^n \widehat{y}^{2} _{1,t} \right)^{-1} \left( \sum_{t=1}^n    \bar{y}_{2,t} \widehat{y}_{1,t}  \right)$, where $\widehat{y}_{1,t} :=  \widehat{\boldsymbol{\beta}}_1^{\prime ivx} \boldsymbol{x}_{1,t-1}$ which implies that it can be expressed as below 
\begin{align}
\widehat{ \delta }^{ols} 
=
\left( \sum_{t=1}^n \widehat{y}^{2} _{1,t} \right)^{-1} \left( \sum_{t=1}^n \bar{y}_{2,t}  \widehat{\boldsymbol{\beta}_1^{\prime ivx} \boldsymbol{x}_{1,t-1} } \right)
=
\left( \sum_{t=1}^n  \widehat{y}^{2} _{1,t} \right)^{-1} \left( \sum_{t=1}^n  \bar{y}_{2,t} \boldsymbol{x}_{1,t-1} \right) \widehat{\boldsymbol{\beta}_1^{\prime ivx} }.
\end{align}
However, $\bar{y}_{2,t} =  \boldsymbol{\beta}_2^{\prime} \bar{\boldsymbol{x}}_{2,t-1}  + \delta \widehat{y}_{1,t} + u_{2,t}$ which implies that 
\begin{align*}
\widehat{ \delta }^{ols}
&=
\left( \sum_{t=1}^n  \widehat{y}^{2} _{1,t}  \right)^{-1} 
\left( \sum_{t=1}^n \bigg[ \boldsymbol{\beta}_2^{\prime} \bar{\boldsymbol{x}}_{2,t-1}  + \delta \widehat{y}_{1,t} + u_{2,t} \bigg] \boldsymbol{x}_{1,t-1}^{\prime} \right) \widehat{\boldsymbol{\beta}_1^{ivx} }
\\
&=
\left( \sum_{t=1}^n \widehat{y}_{1,t} \widehat{y}_{1,t}^{\prime} \right)^{-1} 
\times 
\left\{ \boldsymbol{\beta}_2^{\prime} \left( \sum_{t=1}^n \bar{\boldsymbol{x}}_{2,t-1} \boldsymbol{x}_{1,t-1}^{\prime} \right) + \delta \left( \sum_{t=1}^n  \widehat{y}_{1,t}  \boldsymbol{x}_{1,t-1}^{\prime}  \right) 
\left( \sum_{t=1}^n u_{2,t} \boldsymbol{x}_{1,t-1}^{\prime} \right) \right\} \widehat{\boldsymbol{\beta}_1^{ivx} }
\end{align*}
where 
\begin{align*}
\left( \sum_{t=1}^n  \widehat{y}_{1,t}  \widehat{y}_{1,t}^{\prime} \right)^{-1} 
= 
\left( \sum_{t=1}^n  \widehat{\boldsymbol{\beta}_1^{\prime ivx} \boldsymbol{x}_{1,t-1} \boldsymbol{x}_{1,t-1}^{\prime} \widehat{\boldsymbol{\beta}}_1^{ ivx}  } \right)^{-1} 
=
 \left( \widehat{\boldsymbol{\beta}}_1^{ ivx} \right)^{-1} \left( \sum_{t=1}^n  \boldsymbol{x}_{1,t-1} \boldsymbol{x}_{1,t-1}^{\prime} \right)^{-1} \left( \widehat{\boldsymbol{\beta}}_1^{\prime ivx} \right)^{-1}
\end{align*}
Therefore, we obtain that 
\begin{align*}
\widehat{ \delta }^{ols}
=
\left( \sum_{t=1}^n  \widehat{y}^{2} _{1,t}  \right)^{-1} \boldsymbol{\beta}_2^{\prime} \left( \sum_{t=1}^n \bar{\boldsymbol{x}}_{2,t-1} \boldsymbol{x}_{1,t-1}^{\prime} \right)  \widehat{\boldsymbol{\beta}_1^{ivx} }
+
\left( \sum_{t=1}^n  \widehat{y}^{2} _{1,t} \right)^{-1} \delta \left( \sum_{t=1}^n  \widehat{y}_{1,t}  \boldsymbol{x}_{1,t-1}^{\prime}  \right) \widehat{\boldsymbol{\beta}_1^{ivx} }
+
\left( \sum_{t=1}^n  \widehat{y}^{2} _{1,t} \right)^{-1} \left( \sum_{t=1}^n u_{2,t} \boldsymbol{x}_{1,t-1}^{\prime} \right) \widehat{\boldsymbol{\beta}_1^{ivx} }
\end{align*}
Notice that the second term above becomes: 
\begin{align*}
\left( \sum_{t=1}^n  \widehat{y}^{2} _{1,t} \right)^{-1} \delta \left( \sum_{t=1}^n   \widehat{\boldsymbol{\beta}_1^{\prime ivx} \boldsymbol{x}_{1,t-1} \boldsymbol{x}_{1,t-1}^{\prime} \widehat{\boldsymbol{\beta}}_1^{ ivx}  } \right)
=
\delta \left( \sum_{t=1}^n  \widehat{\boldsymbol{\beta}_1^{\prime ivx} \boldsymbol{x}_{1,t-1} \boldsymbol{x}_{1,t-1}^{\prime} \widehat{\boldsymbol{\beta}}_1^{ ivx}  } \right)^{-1}  \left( \sum_{t=1}^n  \widehat{\boldsymbol{\beta}_1^{\prime ivx} \boldsymbol{x}_{1,t-1} \boldsymbol{x}_{1,t-1}^{\prime} \widehat{\boldsymbol{\beta}}_1^{ ivx}  } \right)
= 
\delta
\end{align*}
Thus, we have that
\begin{align*}
\left( \widehat{ \delta }^{ols}  - \delta \right)
=
\left( \sum_{t=1}^n  \widehat{y}^{2} _{1,t} \right)^{-1} \boldsymbol{\beta}_2^{\prime} \left( \sum_{t=1}^n \bar{\boldsymbol{x}}_{2,t-1} \boldsymbol{x}_{1,t-1}^{\prime} \right)  \widehat{\boldsymbol{\beta}}_1^{ivx} 
+
\left( \sum_{t=1}^n \widehat{y}^{2} _{1,t} \right)^{-1} \left( \sum_{t=1}^n u_{2,t} \boldsymbol{x}_{1,t-1}^{\prime} \right) \widehat{\boldsymbol{\beta}_1^{ivx} }.
\end{align*}

\newpage

Furthermore, since the nonstationary regressors of node 1 are not correlated with the nonstationary regressors of node 2, then the term $\underset{ n \to \infty }{ \mathsf{plim} } \left( \sum_{t=1}^n u_{2,t} \boldsymbol{x}_{1,t-1}^{\prime} \right)  = 0$, which implies that 
\begin{align*}
\left( \widehat{ \delta }^{ols} - \delta \right)
&=
\left( \sum_{t=1}^n \widehat{y}^{2} _{1,t} \right)^{-1} \boldsymbol{\beta}_2^{\prime} \left( \sum_{t=1}^n \bar{\boldsymbol{x}}_{2,t-1} \boldsymbol{x}_{1,t-1}^{\prime} \right) \times  \widehat{\boldsymbol{\beta}_1^{ivx} }
\\
&=
\left( \sum_{t=1}^n  \widehat{y}^{2} _{1,t} \right)^{-1} \boldsymbol{\beta}_2^{\prime} \left(  \mathbb{E} \big[ \boldsymbol{x}_{2,t-1} \boldsymbol{x}_{1,t-1}^{\prime} \big] \right) 
\times \left[ \boldsymbol{\beta}_1 + \left(  \sum_{t=1}^n \bar{\boldsymbol{x}}_{1,t-1} \boldsymbol{z}_{1,t-1}^{\prime} \right)^{-1} \left(  \sum_{t=1}^n \bar{u}_{1,t} \boldsymbol{z}_{1,t-1}^{\prime}  \right) \right]
\\
&=
\left( \sum_{t=1}^n  \widehat{y}^{2} _{1,t} \right)^{-1} \boldsymbol{\beta}_2^{\prime} \left(  \mathbb{E} \big[ \boldsymbol{x}_{2,t-1} \boldsymbol{x}_{1,t-1}^{\prime} \big] \right) \boldsymbol{\beta}_1   
+ \left( \sum_{t=1}^n  \widehat{y}^{2} _{1,t} \right)^{-1} \boldsymbol{\beta}_2^{\prime} \left(  \mathbb{E} \big[ \boldsymbol{x}_{2,t-1} \boldsymbol{x}_{1,t-1}^{\prime} \big] \right)
\\
&\ \ \ \ \ \ \ \ \times 
\left[  \left(  \sum_{t=1}^n \bar{\boldsymbol{x}}_{1,t-1} \boldsymbol{z}_{1,t-1}^{\prime} \right)^{-1} \left(  \sum_{t=1}^n \bar{u}_{1,t} \boldsymbol{z}_{1,t-1}^{\prime}  \right) \right] + o_p(1).
\end{align*}
Since the term $\left( \sum_{t=1}^n \widehat{y}^{2} _{1,t} \right)$ is always positive then we assume that $\underset{ n \to \infty }{ \mathsf{plim} } \left( \sum_{t=1}^n  \widehat{y}^{2} _{1,t} \right)^{-1} = \mathcal{K}$ where $\mathcal{K}$ is some positive constant. Therefore, putting all related expressions together we obtain
\begin{align*}
\label{the.expression2}
n^{ \frac{ 1 + \gamma_z }{2} } \bigg[ \bigg( \tilde{ \boldsymbol{\beta} }_2^{ivx} - \boldsymbol{\beta}_2 \bigg) - \boldsymbol{\beta}_1 \bigg] 
\nonumber
&= 
\bigg\{ \mathcal{K} \boldsymbol{\beta}_2^{\prime} \left(  \mathbb{E} \big[ \boldsymbol{x}_{2,t-1} \boldsymbol{x}_{1,t-1}^{\prime} \big] \right) \boldsymbol{\beta}_1   
+ \mathcal{K} \boldsymbol{\beta}_2^{\prime} \left(  \mathbb{E} \big[ \boldsymbol{x}_{2,t-1} \boldsymbol{x}_{1,t-1}^{\prime} \big] \right) 
\\
& \ \times \mathcal{MN} \bigg( 0, \sigma_{ u_1 }^2 \times \left( \boldsymbol{\Psi}_{u_1u_1}^{-1} \right)^{\prime} \boldsymbol{C}_z  \boldsymbol{V}_{C_z} \boldsymbol{C}_z \boldsymbol{\Psi}_{u_1u_1}^{-1} \bigg)  \bigg\}
\nonumber
\\
& \ \times \mathcal{MN} \bigg( 0, \sigma_{ u_1 }^2 \times \left( \boldsymbol{\Psi}_{u_1u_1}^{-1} \right)^{\prime} \boldsymbol{C}_z  \boldsymbol{V}_{C_z} \boldsymbol{C}_z \boldsymbol{\Psi}_{u_1u_1}^{-1} \bigg)
\nonumber
\\
&+ 
\mathcal{MN} \bigg( 0, \sigma_{ u_2 }^2 \times \left( \boldsymbol{\Psi}_{u_2u_2}^{-1} \right)^{\prime} \boldsymbol{C}_z  \boldsymbol{V}_{C_z} \boldsymbol{C}_z \boldsymbol{\Psi}_{u_2 u_2 }^{-1} \bigg)
\nonumber
\\
&=
\mathcal{K} \boldsymbol{\beta}_2^{\prime} \left(  \mathbb{E} \big[ \boldsymbol{x}_{2,t-1} \boldsymbol{x}_{1,t-1}^{\prime} \big] \right) \boldsymbol{\beta}_1 \times \mathcal{MN} \bigg( 0, \sigma_{ u_1 }^2 \times \left( \boldsymbol{\Psi}_{u_1u_1}^{-1} \right)^{\prime} \boldsymbol{C}_z  \boldsymbol{V}_{C_z} \boldsymbol{C}_z \boldsymbol{\Psi}_{u_1u_1}^{-1} \bigg)
\nonumber
\\
&+
\mathcal{K} \boldsymbol{\beta}_2^{\prime} \left(  \mathbb{E} \big[ \boldsymbol{x}_{2,t-1} \boldsymbol{x}_{1,t-1}^{\prime} \big] \right) \times \bigg\{ \mathcal{MN} \bigg( 0, \sigma_{ u_1 }^2 \times  \left( \boldsymbol{\Psi}_{u_1u_1}^{-1} \right)^{\prime} \boldsymbol{C}_z  \boldsymbol{V}_{C_z} \boldsymbol{C}_z \boldsymbol{\Psi}_{u_1u_1}^{-1} \bigg)  \bigg\}^2
\nonumber
\\
&+
\mathcal{MN} \bigg( 0, \sigma_{ u_2 }^2 \times \left( \boldsymbol{\Psi}_{u_2u_2}^{-1} \right)^{\prime} \boldsymbol{C}_z  \boldsymbol{V}_{C_z} \boldsymbol{C}_z \boldsymbol{\Psi}_{u_2 u_2 }^{-1} \bigg).
\end{align*}
The second term above corresponds to a squared mixed Gaussian distribution which gives a form of a $\chi^2$ distribution. Thus, the limiting distribution of \eqref{the.expression2} is a linear combination of a $\chi^2$ distribution and a mixed Gaussian distribution, which is a Generalized $\chi^2$ distribution. Thus the $y_t$ inherits the properties of $x_t$, through a cointegrating relation, especially our goal is to have a stationary innovation sequence such that $u_t \sim I(0)$. Under the assumption of a consistent estimator for $\hat{\delta}^{ols}$ then we have convergence in probability to zero such that $\left( \hat{\delta}^{ols} - \delta \right) \overset{p}{\to} 0$. Thus, the first term of the expression asymptotically tends to zero (negligible) which implies that the second term converges into a mixed Gaussian distribution.

\newpage

\section{Doubly Corrected Estimation Method: Conditional Quantile Case}
\label{Section4}

\subsection{Estimation based on a Doubly IVX Corrected Procedure}

In this section we consider the double corrected conditional quantile estimation methodology which is the main focus of our research study. The derivations presented in the previous sections were useful to shed light on some challenges we have to overcome to develop a robust framework in the proposed setting of quantile predictive regression models. In terms of estimation approach we focus on the QR estimators under nonstationarity which are adjusted using the IVX filtration. During the first stage of the estimation procedure in the same spirit as several studies in the literature related to the joint estimation of the risk measures of $( \mathsf{VaR}, \mathsf{CoVaR} )$ (see, \cite{AB2016covar}, \cite{hardle2016tenet} and \cite{patton2019dynamic} among others), the $\mathsf{VaR}$ is estimated given a fixed quantile $\tau \in (0,1)$ level but using information on the nonstationary properties of regressors as given by the quantile predictive regression system below
\begin{itemize}

\item[\textbf{Stage 1:}]
\begin{align}
y_{t}^{(1)} &= \alpha_{1} ( \tau ) + \boldsymbol{\beta}_1^{\prime} (\tau) \boldsymbol{x}_{t-1}^{(1)} + u_{t}^{(1)} (\tau), \ \ \ \text{for} \ \ \ t = 1,...,n
\\
\boldsymbol{x}_{t}^{(1)} &= \boldsymbol{R}_{n}^{(1)} \boldsymbol{x}^{(1)}_{t-1} + \boldsymbol{v}^{(1)}_{t}
\end{align}
such that $\alpha_{1} ( \tau ) + \boldsymbol{\beta}^{\prime} (\tau) \boldsymbol{x}_{t-1}^{(1)}$ is the $\tau-$conditional quantile of $y_t$ given $\boldsymbol{x}_{t-1}$, where the unknown parameters $( \alpha_{1} ( \tau ), \boldsymbol{\beta}_1 ( \tau) )$ are estimated using the QR objective function.

\item[\textbf{Stage 2:}]
\begin{align}
y_{t}^{(2)} &= \alpha_{2} ( \tau ) + \boldsymbol{\beta}_2^{\prime} (\tau) \boldsymbol{x}_{t-1}^{(2)} + \delta (\tau) \widehat{y}_{1,t} \left(  \hat{\boldsymbol{\theta}}^{ivx}_1 (\uptau)    \right)   + u_{t}^{(2)} (\tau), \ \ \ \text{for} \ \ \ t = 1,...,n
\\
\boldsymbol{x}_{t}^{(2)} &= \boldsymbol{R}_{n}^{(2)} \boldsymbol{x}^{(2)}_{t-1} + \boldsymbol{v}^{(2)}_{t}
\end{align}
where the parameter of interest corresponds to the estimator $\boldsymbol{\vartheta}^{ivx}_n = \big( \alpha_2 (\tau), \boldsymbol{\beta}_2^{\prime} (\tau), \delta (\tau) \big)$.

\end{itemize}

\begin{remark}
Similar to the remark found in \cite{wang2016conditional}, we point out that it is up to the practitioner to determine the specific model and parameter estimation method, which are the starting point to carry out any subsequent $\mathsf{CoVaR}$ estimation and inference.  In our study, we use the nonstationary quantile predictive regression when obtaining an estimate for $\mathsf{CoVaR}$. Thus, our approach gives a $\sqrt{n^{1 + \delta}}-$consistent estimator for the unknown coefficient of systemic risk based on the IVX filter. The particular optimization function corresponds to the second-stage procedure which gives the model estimates to construct the one-period ahead forecasted $\mathsf{CoVaR}$ risk measure.  However, our approach differs from the estimation method of \cite{AB2016covar} and \cite{hardle2016tenet} since nonstationary time series data are modelled via the LUR parametrization and thus we employ the QR-IVX estimator proposed by  \cite{lee2016predictive} but with the additional generated regressor from the first-stage procedure. 
\end{remark}

\newpage

Next, we focus on the required preliminary theory for establishing the consistency and asymptotic mixed Gaussianity of the doubly IVX corrected estimator $\boldsymbol{\beta}^{ivx} (\uptau)$. Before proceeding with our asymptotic theory analysis, it is worth mentioning some related studies on the aspect of generated regressor in quantile regressions. In particular, in comparison to the framework of \cite{bhattacharya2020quantile}  (see, also \cite{chen2021quantile}), our generated covariate represents a proxy for $\mathsf{VaR}$ under time series nonstationarity. During the second-stage procedure this risk proxy is augmented within a nonstationary quantile predictive regression along with other nonstationary regressors. In addition, to ensure robustness to the abstract degree of persistence we employ the IVX estimator rather than the OLS estimator thereby avoiding the necessity to use the control function approach as a method for correcting size distortions under the null hypothesis of no predictability due to the presence of the nuisance parameter of persistence. Moreover, similar to the illustrative case of the doubly corrected estimation for the conditional mean case we present results for bivariate systems of quantile predictive regression models. 

Our econometric environment can be extended to the multivariate case. In that case, we can develop asymptotic theory and estimation techniques for inference under two large-sample regimes such that $(a)$ with increasing time sample size, $n \to \infty$, and fixed network dimension, denoted by $m$, and $(b)$ with $m \to \infty$ and $n_m \to \infty$, where the temporal size depends on $m$. In other words, the case in which both indices tend to infinity has certain challenges when deriving the asymptotic properties or such complex tail dependent processes.  To the best of our knowledge our study tackles the problem of general asymptotic inference with increasing dimension network time series models under nonstationarity which is a novel aspect to the literature. Suitable Bahadur-based representations for establishing asymptotic theory results are considered in several studies (see, \cite{portnoy2012nearly} and \cite{wang2016conditional}). Within our setting we have the additional complexity that the quantile-based estimators in each stage has to be robustify against the unknown persistence properties that appear due to the LUR parametrization of regressors (e.g., see \cite{lee2016predictive}).

In relation to technical conditions, the notion of stochastic equicontinuity can be employed in nonstationary time series environments.  These conditions are not uncommon in the cointegrating regression literature such as the restricted cointegrated regression model of \cite{saikkonen1995problems} and the constrained least squares estimator proposed by \cite{nagaraj1991estimation} (see,  \cite{moon2002minimum}). Furthermore, as  pointed out by \cite{saikkonen1995problems}, the classical stochastic equicontinuity conditions can be too strong in cointegrated regression models. In particular, due to the stationarity of the generated regressor we shall assume that the parameter space in the nonstationary quantile predictive regression model in the second-stage procedure corresponds to a short-run parameter while the coefficient of the nonstationary regressors corresponds to the long-run parameters of the model. This can be done by  partitioning $\theta = ( \theta_1, \theta_2  )$ and considering a standardization matrix such that 
\begin{align}
N_n ( \theta_1, \delta  ) = \big\{  \boldsymbol{\theta}_1^{*} \in \Theta_1 : \norm{  \boldsymbol{D}_{1n} \left( \boldsymbol{\theta}_1^{*}  - \boldsymbol{\theta}_1 \right)  }  < \delta \big\},
\end{align}

\newpage

where $\boldsymbol{D}_{1n}$ is $( k_1 \times k_1  )$ nonstochastic diagonal matrix whose diagonal elements are positive and increasing functions of $n$. Moreover, these SE conditions are considered to be sufficient conditions for establishing the weak convergence of such estimators while necessary conditions imply the inverse information matrix has eigenvalues bounded away from zero in probability. Therefore, in the case of integrated and cointegrated variables, \cite{saikkonen1995problems} show that the standardized Hessian matrix $\displaystyle  \boldsymbol{D}_{ \boldsymbol{\theta} n  }^{-1}  \frac{ \partial^2 \bar{Q}_n( \boldsymbol{\theta}_0 )   }{   \partial \boldsymbol{\theta} \partial \boldsymbol{\theta}^{\prime} }  \boldsymbol{D}_{ \boldsymbol{\theta} n  }^{-1}$ converges weakly and due to the strict exogeneity assumption, the weak limit is block diagonal between the parameters $\boldsymbol{\theta}_1$, $\psi$ and $\alpha$. The next step would be to show that the standardized Hessian satisfies an appropriate stochastic equicontinuity condition, that along with a suitable consistency result, makes it possible to obtain the limiting distribution. Roughly speaking, OLS-based residuals from the first-stage procedure depend on the nuisance parameter of persistence thus stochastic equicontinuity arguments need to be extended to a suitable probability space that accounts for dynamic misspecification (see,  \cite{kasparis2012dynamic}). 
 
\begin{proposition}
Suppose that the parameter vector of interest is given by the optimization function 
\begin{align}
\widehat{ \boldsymbol{\vartheta} } (\uptau) := \big(  \hat{\alpha} (\uptau), \hat{ \boldsymbol{\beta} }^{\prime}_x (\uptau), \hat{ \boldsymbol{\beta} }^{\prime}_{ \hat{y}  }  (\uptau) \big)^{\prime} \equiv \underset{ \boldsymbol{\vartheta} \in \Theta  }{\mathsf{arg \ min} } \ \sum_{t=1}^n \rho_{\uptau} \big(  y_t - \boldsymbol{X}_t^{\prime} \boldsymbol{\vartheta}  \big)
\end{align}
where $\boldsymbol{X}_t = \big[  1, \boldsymbol{x}_{t-1},  \hat{y}_t  \big]$. Assume that the matrix $\boldsymbol{D}_0$ exists and is positive definite, then as $N,T \to \infty$ for any quantile level $\uptau \in (0,1)$, the following limit results holds 
\begin{align}
\label{limit}
\sqrt{N} \left(  \hat{\boldsymbol{\phi} } ( \cdot ) - \boldsymbol{\phi}_0 ( \cdot  ) \right) \overset{ d}{\to} \boldsymbol{D}_{1NT} (\cdot)^{-1} \mathbb{G}_N \big(  \psi (  \theta_i - \boldsymbol{x}_i^{\prime} \boldsymbol{\phi}_0 ( \cdot  )  )  \big) +o_p(1) \Rightarrow \mathbb{G}_N (\cdot) \in \mathcal{\ell}^{\infty} (\mathcal{T}),
\end{align}
where $\mathbb{G} ( \cdot )$ is a zero mean Gaussian process with covariance kernel $\mathbb{E} \big[ \mathbb{G}(\uptau_1) \mathbb{G}^{\prime}(\uptau_2) \big] $ such that $\mathcal{T} := [ \uptau_1, \uptau_2  ]$. Recall that the weak convergence result holds because the stochastic process is asymptotically tight, and, therefore, converge weakly in $\mathcal{L}^{\infty}$, where  $\mathcal{L}^{\infty}$ denotes the set of all uniformly bounded real functions on $\chi$ equipped with the uniform norm. 

\medskip

A special case to the weak convergence result given by expression \eqref{limit} is to consider the weak convergence of the quantile regression process $\hat{\boldsymbol{\phi} } ( \cdot )$, where for any fixed quantile level $\uptau \in (0,1)$ it holds that as $ (N,T) \to \infty$, 
\begin{align*}
\sqrt{N} \left(  \hat{\boldsymbol{\phi} } ( \uptau ) - \boldsymbol{\phi}_0 ( \uptau ) \right) \overset{  d}{\to} \mathcal{N} \left(  0, \uptau ( 1 - \uptau ) \boldsymbol{D}_{1} \boldsymbol{D}_0^{-1} \boldsymbol{D}_{1}\right).
\end{align*}
Define with $\mathbb{G} (\uptau)$ is a zero mean Gaussian process with covariance kernel 
\begin{align}
\mathbb{E} \big[ \mathbb{G} (\uptau) \mathbb{G}^{\prime} (\uptau^{\prime})  \big] = \boldsymbol{\Xi} (\uptau)^{-1}(\uptau) S(\uptau, \uptau^{\prime} )  \boldsymbol{\Xi} (\uptau)^{-1}(\uptau) \boldsymbol{D}_0^{-1}.
\end{align}
\end{proposition}

\newpage

\begin{remark}
Recall that a key condition for estimating the quantile-dependent cointegrating vector is that there exists a zero correlation between regressors and errors (see,  \cite{xiao2009quantile}).   Moreover, cases where estimators are derived under the assumption large $N,T$ asymptotics include the frameworks of \cite{zhu2019network}, \cite{chen2022two} and \cite{feng2024estimation} among others. Nevertheless, the above derivations provide a good example on the possible system representation we can employ in the case we consider the generated regressors that correspond to the estimated $\mathsf{VaR}$ (under time series nonstationarity) across the cross-section.  The above example which corresponds to the ARDL process as in the study of \cite{cho2015quantile} is presented also by  \cite{katsouris2023quantile}. 
\end{remark}

\begin{example}[see, \cite{cho2015quantile}]
Consider the ARDL process denoted by 
\begin{align*}
Y_t = \alpha + \sum_{j=1}^p \phi_j Y_{t-j} + \sum_{j=0}^q \boldsymbol{\theta}_j^{\prime} \boldsymbol{X}_{t-j} + U_t,
\end{align*}
where $ \boldsymbol{X}_t \in \mathbb{R}^k$ is a $k-$dimensional integrated process of stationary and ergodic processes. Moreover, let the $\tau-$quantile of $Y_t$ conditional on the $\sigma-$algebra $\mathcal{F}_{t-1}$ denoted by $Q_{ y_t } ( \tau | \mathcal{F}_{t-1}  )$ be
\begin{align*}
Q_{ y_t } ( \tau | \mathcal{F}_{t-1}  ) = \alpha (\tau) + \sum_{j=1}^p \phi_j (\tau)  Y_{t-j} + \sum_{j=0}^q \boldsymbol{\theta}_j^{\prime} (\tau)  \boldsymbol{X}_{t-j} 
\end{align*} 
We consider that $Y_t$ is represented as a quantile autoregressive distributed lag (QARDL) process 
\begin{align}
Q_{ y_t } ( \tau | \mathcal{F}_{t-1}  ) = \alpha (\tau) + \sum_{j=1}^p \phi_j (\tau)  Y_{t-j} + \sum_{j=0}^q \boldsymbol{\theta}_j^{\prime} (\tau)  \boldsymbol{X}_{t-j} + U_t(\tau)
\end{align} 
where $U_t(\tau) = Y_t - Q_{ y_t } ( \tau | \mathcal{F}_{t-1}  )$.  We reformulate the above expression so that it can capture the long-run relationship between $Y_t$ and $\boldsymbol{X}_t$ using the following long-run quantile process
\begin{align}
Y_t = \mu (\tau) + \boldsymbol{X}_t^{\prime} \boldsymbol{\beta} (\tau) + \mathcal{R}_t ( \tau )
\end{align}
We explain below the steps and the coefficients to obtain the above reformulation. 
\begin{itemize}

\item[\textbf{Step 1.}] We decompose the term $\textcolor{blue}{ \sum_{j=0}^q \boldsymbol{\theta}_j^{\prime} (\tau)  \boldsymbol{X}_{t-j} } \equiv \textcolor{blue}{ \boldsymbol{X}_t^{\prime} \boldsymbol{\gamma}(\tau) + \sum_{j=0}^{q-1} \boldsymbol{Z}_{t-j}^{\prime} \boldsymbol{\delta} (\uptau ) }$ which implies that
\begin{align*}
Y_t = \alpha (\tau) + \sum_{j=1}^p \phi_j (\tau)  Y_{t-j}  +  \underbrace{ \textcolor{blue}{ \boldsymbol{X}_t^{\prime} \boldsymbol{\gamma}(\tau) + \sum_{j=0}^{q-1} \boldsymbol{Z}_{t-j}^{\prime} \boldsymbol{\delta} (\uptau ) } }_{  \sum_{j=0}^q \boldsymbol{\theta}_j^{\prime} (\tau)  \boldsymbol{X}_{t-j}  }  + U_t(\tau)
\end{align*}
where $\boldsymbol{\gamma} (\tau) := \sum_{j=0}^q \boldsymbol{\theta}_j (\tau)$, $\boldsymbol{\delta} (\tau) :=  - \sum_{ i = j+1 }^q \boldsymbol{\theta}_i (\tau)$ and $\boldsymbol{Z}_t = \Delta \boldsymbol{X}_{t}$.

\newpage

\item[\textbf{Step 2.}]  Suppose that 
\begin{align*}
Y_t = \alpha (\tau) + \sum_{j=1}^p \phi_j (\tau)  Y_{t-j}  +  \underbrace{ \textcolor{blue}{ \boldsymbol{X}_t^{\prime} \boldsymbol{\gamma}(\tau) + \sum_{j=0}^{q-1} \boldsymbol{Z}_{t-j}^{\prime} \boldsymbol{\delta} (\uptau ) } }_{  \sum_{j=0}^q \boldsymbol{\theta}_j^{\prime} (\tau)  \boldsymbol{X}_{t-j}  }  + U_t(\tau) \equiv \mu (\tau) + \boldsymbol{X}_t^{\prime} \boldsymbol{\beta} (\tau) + \mathcal{R}_t ( \tau )
\end{align*}
where it holds that
\begin{align}
\boldsymbol{\beta} (\tau)  := \frac{  \boldsymbol{\gamma}(\tau)    }{ \displaystyle \left(  1 -  \sum_{j=1}^p \phi_j (\tau)  \right) }, \ \ \mu(\tau) :=  \frac{ \alpha(\tau)    }{ \displaystyle \left(  1 -  \sum_{j=1}^p \phi_j (\tau)  \right) }
\end{align}
and the remainder term can be written such that 
\begin{align*}
R_t (\tau) := \sum_{j=0}^{\infty} \boldsymbol{Z}^{\prime}_{t-j} \boldsymbol{\xi}_j (\tau) +  \sum_{j=0}^{\infty} \eta_j (\tau) U_{t-j} (\tau)
\end{align*}
where each $\boldsymbol{\xi}_j (\tau) $ corresponds to the "upper" unbounded sum of some positive sequence of coefficients $\boldsymbol{\pi}_{\ell} (\tau)$ such that $\boldsymbol{\xi}_j (\tau) \equiv \sum_{\ell = j +1}^{ \infty } \boldsymbol{\pi}_{\ell} (\tau)$ where the set of coefficients $\big\{   \boldsymbol{\pi}_{0} (\tau)      , \boldsymbol{\pi}_{1} (\tau),...   \big\}$ and $\big\{ \eta_0(\tau), \eta_1(\tau),...  \big\}$ which implies that
\begin{align*}
\sum_{j=0}^{\infty} \eta_j  (\tau) L^j &\equiv \left( 1 - \sum_{j=1}^p \phi_j ( \tau) L^j  \right)^{-1}, 
\\
\sum_{j=0}^{\infty} \boldsymbol{\xi}_j  (\tau) L^j  &\equiv  ( 1 - L )^{-1} \left(  \frac{ \displaystyle  \sum_{j=0}^q \boldsymbol{\theta}_j (\tau) L^j  }{ \displaystyle  1 - \sum_{j=1}^p \phi_j  (\tau)  L^j   }  -     \frac{ \displaystyle     \sum_{j=0}^q \boldsymbol{\theta}_j (\tau)  }{ \displaystyle  1 - \sum_{j=1}^p \phi_j  (\tau)   }   \right). 
\end{align*}
Our main interest lie in developing the estimation theory for the long-run parameter $\boldsymbol{\beta} (\tau)$ which is a function of $\boldsymbol{\gamma} (\tau)$ and $\boldsymbol{\phi}_j := \big(  \phi_1 (\tau),...,  \phi_p (\tau)  \big)^{\prime}$.   
\end{itemize}

\end{example}

In a similar spirit as in the illustrative example above in which the an invertible representation of the ARDL exists, we consider that our proposed econometric specification which corresponds to an augmented quantile predictive regression model with generated regressors, could be modified accordingly using equivalent type of representations. We leave some derivations for a future study. On the other hand, we should emphasize again that the generated covariates we consider in our setting are allows to be possibly related to different observable variables in different functional forms. Nevertheless, these GR equations can be estimated separately (i.e., estimation of the main diagonal of the risk matrix in \cite{katsouris2021optimal}), and thus any dependence between the GRs can be captured by the variance covariance matrix (this separation of the proxy risk measures of  $\mathsf{VaR}$ is also reflected in the novel two-step estimation approach proposed by \cite{katsouris2021optimal}).


\newpage

\subsection{Misspecification Analysis of $\mathsf{VaR}$-$\mathsf{CoVaR}$ risk matrix}

In order to obtain some further insights in the implications of the theoretical aspects discussed in this paper, we briefly discuss the extension of the framework proposed by \cite{katsouris2023statistical} into the econometric environment of this study.  Suppose that $\boldsymbol{y}_{it}$ is generated by the true adjacency matrix $A$ (e.g., consider the novel $\mathsf{VaR}-\Delta\mathsf{CoVaR}$ risk matrix proposed by \cite{katsouris2021optimal} as the network adjacency matrix). Related studies from the literature of network regression estimation and inference include \cite{zhu2017network}, \cite{zhu2019network} and \cite{zhu2020grouped} among others).  

Roughly speaking, the consistency result for the population counterpart of the risk matrix might not be held when the adjacency matrix $A$ is misspecified to be $A^{\star} = \left( a_{ij}^{\star} \right)$. Denote with $W^{\star}$ to be the row-normalized $A^{\star}$ and $X_{it}^{\star}$ to be the matrix of model regressors. Then, the estimator is 
\begin{align}
\hat{ \boldsymbol{\beta} }^{\star} ( \uptau ) = \underset{ \boldsymbol{\beta}^{\star} }{  \mathsf{arg \ min}  } \sum_{ i=1 }^N \sum_{ t = 1}^T \uprho_{ \uptau } \left( y_{it} - \boldsymbol{X}_{(i)t-1}^{\star \top} \boldsymbol{\beta}^{\star} ( \uptau ) \right).
\end{align}
Define with
\begin{align}
\widehat{ \boldsymbol{ D } }_0 &= \frac{1}{NT} \sum_{i=1}^N \sum_{t=1}^T \boldsymbol{X}_{(i)t-1}^{\star} \boldsymbol{X}_{(i)t-1}^{\star \top},  
\\
\widehat{ \boldsymbol{ D } }_1 ( \uptau) &= \frac{1}{NT} \sum_{i=1}^N \sum_{t=1}^T f_{it} \left(  \boldsymbol{X}_{(i)t-1}^{\star \top} \boldsymbol{\beta}^{\star} ( \uptau) \right) \boldsymbol{X}_{(i)t-1}^{\star} \boldsymbol{X}_{(i)t-1}^{\star \top}.
\end{align}

\begin{remark}
Notice that for the VaR-CoVaR risk matrix, misspecification has a different interpretation than for the usual binary or weighted adjacency matrix for a graph. The reason is that our risk matrix is constructed based on tail forecasts using a regression-based methodology that incorporates the persistence properties of the regressors (see, \cite{katsouris2021optimal, katsouris2023statistical}).
\end{remark}
The conditions associated with the misspecified coefficients are listed below:  
\begin{assumption}
\textcolor{blue}{ \textit{(Eigenvalue-Bound)} } 
Suppose that $\widehat{ \boldsymbol{ D } }_1 ( \uptau) \overset{ p }{ \to } \boldsymbol{ D }_1 (\uptau)$  as $ \mathsf{min} \left\{ N, T \right\} \to \infty$ for any $\uptau \in (0,1)$, where $\boldsymbol{ D }_1 (\uptau) \in \mathbb{R}^{ N \times N}$ is a positive definite matrix. Then, there exists positive constants, $0 < c_1 < c_2 < \infty$ such that $c_1 \leq \lambda_{ \mathsf{min} } \left(  \boldsymbol{ D }_1 (\uptau)  \right) \leq \lambda_{ \mathsf{max} } \left(  \boldsymbol{ D }_1 (\uptau) \right) \leq c_2$ for any fixed $\uptau \in (0,1)$. 
\end{assumption}

\begin{assumption}
\textcolor{blue}{ \textit{(Monotonicity)} } 
It is assumed that $\boldsymbol{X}_{(i)t-1}^{\star \top} \boldsymbol{\beta}^{\star}(\uptau)$ such that $1 \leq i \leq N$ and $1 \leq t \leq T$ is a monotone increasing function with respect to $\uptau \in (0,1)$. 
\end{assumption}

\begin{corollary}
Define with $d \left( W^{\star}, W \right) = \displaystyle \sum_{i,t} \left| w_{ij}^{\star} - w_{ij} \right|$ to be the total amount of misspecification of $W$. Then, we have that 
\begin{align}
\label{beta.asympt}
\left( \hat{\boldsymbol{\beta}}^{\star}( \uptau ) - \boldsymbol{\beta}^{\star} ( \uptau ) \right) = - \frac{1}{NT}  \left\{ \boldsymbol{D}_1 ( \uptau ) \right\}^{-1} \sum_{i=1}^N \sum_{t=1}^T \boldsymbol{X}_{(i)t-1}^{\star} \psi_{\uptau} \left( \boldsymbol{ \mathcal{E}}_{it \uptau} \right) + \mathcal{R}_{NT} \left( \uptau \right),
\end{align}
with $\mathsf{sup}_{\uptau} \norm{ R_{NT} (\uptau) } = o_p \left( \left( NT \right)^{- 1 / 2} \right)$.
\end{corollary}

\newpage 

Denote with 
\begin{align*}
\boldsymbol{ \mathcal{E}}_{(i)t }(\uptau) := \left( y_{(i)t} - \boldsymbol{X}_{(i)t-1}^{\star \prime} \boldsymbol{\beta}^{\star} ( \uptau ) \right) \ \ \ \text{and} \ \ \ \mathbb{S}_{in}  \big(  \boldsymbol{\beta}^{\star} ( \uptau )  \big) := \frac{1}{n} \sum_{t=1}^n \psi_{\uptau} \left( \boldsymbol{ \mathcal{E}}_{(i)t}  (\uptau) \right) \boldsymbol{X}_{(i)t-1}^{\star}. 
\end{align*}

\medskip

Moreover, assuming that the amount of misspecification has a bounded convergence rate such that $d \left( W^{\star}, W \right) = o_p \left( \sqrt{N / T} \right)$, then we have that  $\left( \hat{\boldsymbol{\beta}}^{\star}( \uptau ) - \boldsymbol{\beta}^{\star} ( \uptau ) \right) = o_p \left( \sqrt{N / T} \right)$. Consequently, it can be shown that when the misspecification amount is under control, that is, $d \left( W^{\star}, W \right) = \mathcal{O}_p \left( \sqrt{N / T} \right)$, the resulting estimator is still $\sqrt{NT}-$consistent. 
\begin{proof}
We divide the proof into two parts. In the first part of the proof, we show the representation given by expression \eqref{beta.asympt} while in the second part we show the consistency result for  $\hat{\boldsymbol{\beta}}^{\star}( \uptau )$. Denote with $\Delta^{\star} \left( \uptau \right) = \hat{\boldsymbol{\beta}}^{\star}( \uptau ) - \boldsymbol{\beta}^{\star} ( \uptau )$. Then, we have that 
\begin{align}
\label{min.rho}
\uprho_{ \uptau } \left( y_{it} - \boldsymbol{X}_{(i)t-1}^{\star \top} \hat{ \boldsymbol{\beta}}^{\star} ( \uptau ) \right) = \uprho_{ \uptau } \left(  \boldsymbol{ \mathcal{E} }_{it \uptau} -   \boldsymbol{X}_{(i)t-1}^{\star} \Delta^{\star} \left( \uptau \right) \right).
\end{align}
where $\boldsymbol{ \mathcal{E} }_{it \uptau} = y_{it} - \boldsymbol{X}_{(i)t-1}^{\star \top} \boldsymbol{\beta}^{\star} ( \uptau )$. Then, the minimization of \eqref{min.rho} is equivalent to minimizing for a fixed $\uptau \in (0,1)$ such that 
\begin{align}
\mathcal{Z}_{NT}^{o} \left( \Delta , \uptau \right) = \left( NT \right)^{-1} \sum_{i=1}^N  \sum_{t=1}^T \left\{ \uprho_{ \uptau } \left( V_{it \uptau }  - \left( NT \right)^{- 1/ 2} X^{\top}_{i(t-1)} \Delta \right) - \uprho_{ \uptau } \left( V_{it \uptau } \right) \right\}.
\end{align}
\end{proof}
A process $W_{NT} \left( \uptau \right)$ is said to be tight if and only if for any $\delta > 0$ there exists a compact set $E$ with
\begin{align}
\underset{ \uptau \in E }{ \mathsf{sup} } \ \mathbb{P} \big(  W_{NT} \left( \uptau \right) \in E \big) > 1 - \delta.
\end{align}

\begin{proof}
We define with 
\begin{align}
\psi_1 (D) = - \frac{1}{ \sqrt{NT} } \sum_{i,t} \big[ \psi_{\uptau_2} \left( V_{it \uptau_2 } \right) - \psi_{\uptau_1} \left( V_{it \uptau_1 } \right) \big]
\end{align}
for any interval $D = ( \uptau_1, \uptau_2  ]$. Therefore, to show the tightness property, we adopt Theorem 15.6 from \cite{billingsley1968convergence} and prove a sufficient Chentsov-Billingsley type of inequality (see, \cite{bickel1971convergence})  as follows. 
\end{proof}

\begin{lemma}
For any two intervals $D_1 = ( \uptau_1, \uptau_2 ]$ and $D_2 = ( \uptau_2, \uptau_3 ]$ we have that 
\begin{align}
\mathbb{E} \left[ \big\{ \eta^{\top} \xi_1 ( D_1 ) \big\}^2 \big\{ \eta^{\top} \xi_1 ( D_2 ) \big\}^2 \right] \leq C (\uptau_3 - \uptau_1 ), 
\end{align}
where $C$ is a finite positive constant (see,  Appendix of \cite{zhu2019network}).  
\end{lemma}

\newpage

Moreover, we denote with $\zeta_{it} \left( \uptau,  \dot{ \uptau} \right) :=  \psi_{ \uptau } \left( V_{it \uptau } \right) -  \psi_{ \dot{ \uptau} } \left( V_{it \dot{ \uptau} } \right) $ for two quantile levels, $\uptau$ and $\dot{ \uptau}$. Then, since $\mathbb{E} \left[ \zeta_{it} \left( \uptau,  \dot{ \uptau} \right) | \mathcal{F}_{t-1} \right] = 0$, using the Cauchy's inequality we can derive an appropriate bound
\begin{align*}
&\mathbb{E} \left[ \left( \sum_{i,t} \eta^{\top}  X_{i(t-1)} \zeta_{it} \left( \uptau_1, \uptau_2 \right)\right)^2 \left( \sum_{i,t} \eta^{\top}  X_{i(t-1)} \zeta_{it} \left( \uptau_2, \uptau_3 \right)\right)^2 \right] \leq 
\\
&\ \ \ \ \ \ \sqrt{  \mathbb{E} \left[ \left( \sum_{i,t} \eta^{\top}  X_{i(t-1)} \zeta_{it} \left( \uptau_1, \uptau_2 \right)\right)^4 \right] } \sqrt{ \mathbb{E} \left[ \left( \sum_{i,t} \eta^{\top}  X_{i(t-1)} \zeta_{it} \left( \uptau_3, \uptau_4 \right)\right)^4 \right] }. 
\end{align*} 
Therefore, it holds that 
\begin{align*}
\mathbb{E} \bigg[ \bigg( \eta^{\top}  X_{i_1( t_1-1 )} \zeta_{i_1 t_1 } \left( \uptau , \dot{ \uptau} \right) \bigg) \bigg( \eta^{\top}  X_{i_2( t_2-1 )} \zeta_{i_2 t_2 } \left( \uptau , \dot{ \uptau} \right) \bigg)  \bigg( \eta^{\top}  X_{i_3( t_3-1 )} \zeta_{i_3 t_3 } \left( \uptau , \dot{ \uptau} \right) \bigg)  \bigg( \eta^{\top}  X_{i_4( t_4-1 )} \zeta_{i_4 t_4 } \left( \uptau , \dot{ \uptau} \right) \bigg) \bigg]
\end{align*}
is non-zero only if (a) $i_1 = i_2, t_1 = t_2$ and $i_3 = i_4 \neq i_1, t_3 = t_4 \neq t_1$ or (b) $i_1 = i_2 = i_3 = i_4$ and $t_1 = t_2 = t_3 = t_4$.  Then, it can be shown that 
\begin{align*}
&\left( NT \right)^{-2} \mathbb{E} \left[ \left( \sum_{i,t} \eta^{\top}  X_{i(t-1)} \zeta_{it} \left( \uptau_1, \uptau_2 \right)\right)^4  \right] 
=
\nonumber
\\
&\ \ \ \ \ \left( NT \right)^{-2} \left\{ \sum_{i,t} \mathbb{E} \left[ \left(  \eta^{\top}  X_{i(t-1)} \right)^2 \zeta_{it}^2 \left( \uptau_1, \uptau_2 \right)  \right] \right\}^2
+ 
\left( NT \right)^{-2} \left\{ \sum_{i,t} \mathbb{E} \left[ \left(  \eta^{\top}  X_{i(t-1)} \right)^4 \zeta_{it}^4 \left( \uptau_1, \uptau_2 \right) \right] \right\}.
\end{align*}
Notice that we have that $\mathbb{E} \left( \eta^{\top} X_{it}  \right)^2 = \mathcal{O}_p(1)$ and  $\mathbb{E} \left( \eta^{\top} X_{it}  \right)^4 = \mathcal{O}_p(1)$.  Moreover, it can be easily verified that $\mathbb{E} \left[  \zeta_{it}^2 \left( \uptau_1, \uptau_2 \right) \right] \leq \uptau_1 - \uptau_2$ and $\mathbb{E} \left[  \zeta_{it}^4 \left( \uptau_1, \uptau_2 \right) \right] \leq \uptau_1 - \uptau_2$. By combining the aforementioned results together we obtain 
\begin{align}
\mathbb{E} \left[ \big\{ \eta^{\top} \xi_1 ( D_1 ) \big\}^2 \big\{ \eta^{\top} \xi_1 ( D_2 ) \big\}^2 \right] \leq C (\uptau_2 - \uptau_1 )(\uptau_3 - \uptau_2 ) \leq C \left| \uptau_3 - \uptau_1  \right|, 
\end{align}
for some positive constant $C$.  We can then conclude that the stochastic quantity $\xi_1 ( \uptau)$ converge weakly to a $( q + 1 )-$dimensional Brownian Bridge. 


\newpage 

\begin{wrap}

\section{Monte Carlo Simulations}

\subsection{Empirical Size}

\begin{itemize}

\item As can be seen, all of the tests are rather well sized for large sample sizes, and the distortions are not very severe for the smallest size either. It can be concluded that the asymptotic null distributions of the test statistics are reasonably good approximations to the unknown finite-sample distributions for $n \geq 1000$. 

\item To complete the experiment, we also investigate the effect of having more persistence in the data-generating process.  

\item To explore the effects of misspecification of the error distribution on the finite-sample properties of our test statistics, we repeated the simulations with different error distributions. 

\item Moreover, notice that unconditional heterogeneity in the innovation sequences can manifest when considering test statistics over a compact set $\uptau \in \Lambda$. 

\item  Another major concern is a comparison between the asymptotic variances of the QR-IVX estimators and the GQR-IVX which accounts for the generated regressor and thus the estimation error from the first stage estimation procedure under time series nonstationarity. 

\end{itemize}

\subsection{Power Simulations}

\begin{itemize}

\item To evaluate the performance of our test statistics, we perform a power comparison to a benchmark test. Moreover, we employ our proposed subvector test statistic that corresponds to the coefficient of the generated regressor. Under the alternative we consider local power functions that capture deviations from the zero coefficient of the systemic risk under the null as a function of the sample size. Our sample sizes are $n \in \left\{ 250, 500 \right\}$ and we use $B = 1,000$ replications. 

\item  Notice that in the case of a subvector testing, by the usual invariance argument of Wald (1943), the joint Wald statistic for testing the null hypothesis of interest, is equal to the sum of the Wald statistic for testing $\delta = 0$ and the statistic for testing $\beta = 0$. Furthermore, alternative combinations of the two statistics that place different weights on each of the two components can  be considered in order to direct power to specific alternatives. 

\item The power performance of our test statistic is also uniformly better than other competing methods under different scenarios when the predictive regressor is highly persistent. A well-known stylized feature in the literature is that highly persistent predictors exhibit strong conditional heteroscedasticity  and endogeneity.

\end{itemize}

\end{wrap}

\newpage 

\section{Conclusion}
\label{Section5}

We study estimation and inference for nonstationary quantile predictive regression models when a generated covariate is required for the purposed of estimating pairs of risk measures under time series nontationarity. In We discuss the main estimation steps of the proposed modeling methodology and provide an  asymptotic theory analysis for the associated quantile-dependent model estimators under regressors nonstationarity which is captured via the LUR parametrization.   

We propose a two-stage estimation procedure, where the generated regressor that corresponds to the one-step ahead prediction of $\mathsf{VaR}$ based on persistent data is generated in the first stage of the estimation procedure. In the second-stage of the estimation procedure the two-stage estimator is employed to construct the one-step ahead prediction of $\mathsf{CoVaR}$ based on the estimated $\mathsf{VaR}$ and another set of persistent data not necessarily identical to the first-stage estimation procedure. We investigate the asymptotic properties of the two-stage estimator based on quantile predictive regression models. Moreover, we show that accounting for the presence of generated regressors when obtaining forecasts of the pair $( \mathsf{VaR}, \mathsf{CoVaR})$, based on the proposed two-stage estimation procedure, results in improved statistical  testing after adjusting the variance-covariance matrix to account for the estimation error in the first-stage estimation procedure. 

In a further study we consider  the econometric identification and estimation of a novel SUR representation that corresponds to the system-specific quantile predictive regressions as proposed by \cite{katsouris2021optimal, katsouris2023quantile, katsouris2023statistical}. Specifically, modelling tail interdependencies in a network using the SUR representation provides efficiency gains and permits to implement relevant statistical inference methods especially when regressors are assumed to be generated from a LUR process.

\paragraph{Conflicts of interest}

The author declares that there are no known conflicts of interest.

\paragraph{Data availability}

No data was used for the research described in this article.

\paragraph{Acknowledgements}

This article is based on the first chapter of my doctoral thesis which was submitted to the University of Southampton (see, \cite{katsouris2021optimal}, \cite{katsouris2023statistical}, \cite{katsouris2023quantile}). I wish to thank Professor Dr. Jose Olmo,  Professor Dr. Tassos Magdalinos and Dr. Jayeeta Bhattacharya from the Department of Economics, University of Southampton for helpful conversations and comments that significantly improve this paper as well as  Professor Dr. Markku Lanne and  Professor Dr. Mika Meitz from the Faculty of Social Sciences, University of Helsinki. I also sincerely thank Dr. Ji Hyung Lee for generously sharing the replication code for the IVX-QR procedure which has been adapted in this article. Financial support from the Research Council of Finland (grant 347986) is gratefully acknowledged.  All remaining errors are my own responsibility.

\newpage

\section{Appendix}

\subsection{Appendix A: Auxiliary Results}

\begin{theorem} (see Theorem 2 in \cite{kato2009asymptotics})
Suppose $f_n ( x , \uptau )$ with $n \geq 1$ are convex in $x$ for each $\uptau$ and bounded in $\uptau$ for each $x$. Moreover, let $g_n( x, \uptau ) = - x^{\top} W_n ( \uptau ) + \frac{1}{2} x^{\top} Q( \uptau ) x$, where $\left\{ W_n(.) \right\}$ is a sequence of bounded stochastic processes and $Q( \uptau )$ is a $d \times d$ non-stochastic symmetric positive definite matrix for each $\uptau$. Suppose that the maximum eigenvalue of $Q( \uptau )$ is bounded from above and the minimum eigenvalue of $Q( \uptau )$ is bounded away from 0 over $\uptau \in \mathcal{B}$. If it holds that 
\begin{align}
\underset{ \uptau \in \mathcal{B} }{ \mathsf{sup} } \left| f_n( x, \uptau ) -  g_n( x, \uptau ) \right| \overset{ p }{ \to } 0, 
\end{align}
for each $x$ and of for every $\eta > 0$, there exists a constant $M > 0$ such that 
\begin{align}
\underset{ n \to \infty }{ \mathsf{lim \ sup} } \ \mathbb{P} \left(  \underset{ \uptau \in \mathcal{B} }{ \mathsf{sup} } \norm{ W_n ( \uptau ) }  > M \right) \leq \eta, 
\end{align}
then $x_n ( \uptau ) = \left\{ Q ( \uptau ) \right\}^{-1} W_n ( \uptau ) + r_n ( \uptau )$, where $\mathsf{sup}_{ \uptau \in \mathcal{B} } \norm{ r_n( \uptau ) } = o_p(1)$. 
\end{theorem}
\begin{proof}
Let $y_n( \uptau ) = \left\{ Q ( \uptau ) \right\}^{-1} W_n ( \uptau )$, which is the unique minimum point of $g_n( . , \uptau)$ for each $\uptau$. Then, it can be shown that 
\begin{align}
g_n ( x, \uptau ) - g_n \left( y_n( \uptau ), \uptau \right) 
= 
\frac{1}{2} \left( x - y_n ( \uptau ) \right)^{\top} Q_n( \uptau )  \left( x - y_n ( \uptau ) \right)
\geq c \norm{ x -  y_n ( \uptau )  },
\end{align}
for some constant $c > 0$. Denote with 
\begin{align}
\Delta_n = \underset{ \uptau \in \mathcal{B} }{ \mathsf{sup} } \  \underset{ x : \norm{ x - y_n( \uptau ) } \leq \delta  }{ \mathsf{sup} }  \left| f_n( x, \uptau ) - g_n ( x , \uptau) \right|.
\end{align}
Therefore, letting $r_n ( \uptau ) = x_n ( \uptau ) - y_n ( \uptau )$, we have that $\mathbb{P}^{*} \left( \underset{ \uptau \in \mathcal{B} }{ \mathsf{sup} } \norm{ r_n ( \uptau ) } > \delta \right) \leq \mathbb{P}^{*} \left( \Delta_n \geq ( c \delta^2 ) / 2 \right)$. So it suffices to show that $\Delta_n \overset{ p }{ \to } 0$. Let $\eta > 0$ be an arbitrary positive constant. Take $M > 0$ such that the following bound holds
\begin{align}
\underset{ n \to \infty }{ \mathsf{ lim \ sup } } \ \mathbb{P}^{*} \left( \underset{ \uptau \in \mathcal{B} }{ \mathsf{sup} } \ \norm{ y_n( \uptau ) }  > M \right) \leq \eta.
\end{align}
Define $K := \left\{ x : \norm{ x - y } \leq \delta , \norm{y} \leq M \right\} = \left\{ x : \norm{x} \leq \delta + M \right\}$. Then, for every $\epsilon > 0$, 
\begin{align}
\mathbb{P}^{*} \left( \Delta_n > \epsilon \right) \leq \mathbb{P}^{*} \left( \underset{ \uptau \in \mathcal{B} }{ \mathsf{sup} } \    \underset{ x \in K }{ \mathsf{sup} } \ \big| f_n ( x, \uptau ) -   g_n ( x, \uptau ) \big| > \epsilon \right) + \mathbb{P}^{*} \left(  \underset{ \uptau \in \mathcal{B} }{ \mathsf{sup} } \norm{ y_n( \uptau ) } > M \right). 
\end{align}

\newpage

Since by Lemma 1, we have that 
\begin{align}
\underset{ \uptau \in \mathcal{B} }{ \mathsf{sup} } \    \underset{ x \in K }{ \mathsf{sup} } \ \big| f_n ( x, \uptau ) -   g_n ( x, \uptau ) \big| = \underset{ \uptau \in \mathcal{B} }{ \mathsf{sup} } \    \underset{ x \in K }{ \mathsf{sup} } \ \big| f_n ( x, \uptau ) + x^{\top} W_n( \uptau ) - \frac{1}{2} x^{\top} Q( \uptau ) x    \big| \overset{ p }{ \to } 0,
\end{align}
We can conclude that $\underset{ n \to \infty }{ \mathsf{lim \ sup} } \mathbb{P}^{*} \left( \Delta_n > \epsilon \right) \leq \eta$. Since $\eta > 0$ is arbitrary, then the proof ends. 
\end{proof}
The theoretical results presented by \cite{kato2009asymptotics} are useful when considering equivalent maximization functionals which can be suitably modified when establishing the asymptotic theory of sample moments in the case of nonstationary quantile predictive regressions (e.g., see, \cite{cai2023new}). 

Denote the conditional CDF of $Y$ given $X = x$ by $F_{Y|X} (\cdot | x)$ and its conditional quantile at $\uptau \in (0,1)$ by $Q ( \uptau | x )$, such that $Q ( \uptau | x ) = F_{Y|X}^{-1} ( \uptau | x ) = \mathsf{inf} \big\{ s : F_{Y|X} ( s | x ) \geq \uptau  \big\}$. Suppose that the quantile $Q( \uptau | x )$ is modelled as a general nonlinear function of $x$ and $\uptau$. We fix $x$ and treat $Q( \uptau | x )$ as a stochastic process in $\uptau$, where $\uptau \in \mathcal{T} = [ \lambda_1, \lambda_2 ]$ with $0 < \lambda_1 \leq \lambda_2 < 1$. 

Furthermore, it can be shown that $\hat{\beta} ( \uptau ) - \beta ( \uptau ) = O_p( n^{- 1 / 2} )$ uniformly in $\uptau \in [ \alpha, 1 - \alpha ]$ for any $\alpha \in (0, 1/ 2 )$. Then, the computational property of regression quantile processes is employed to derive the uniform asymptotic representation of $n^{1 / 2} \left( \hat{\beta} ( \uptau ) - \beta ( \uptau ) \right)$. Notice also that the functional class is P-Donsker for any $B \subset \mathbb{R}^p$
\begin{align}
\big\{ g(y,x) = \big( \uptau - I ( y \leq x^{\top} \beta ) \big) x_j , \beta \in B, \uptau \in [ \alpha, 1 - \alpha ], 1 \leq j \leq p  \big\} 
\end{align}
In other words,  a key quantity for establishing asymptotic theory and inferential methods is the regression-quantile process  such that $\boldsymbol{M}_n ( \uptau ) \equiv  \left\{  \sqrt{n} \left(  \widehat{\boldsymbol{\beta}}_n( \uptau ) - \boldsymbol{\beta} ( \uptau )  \right) \right\}$ for some $\uptau \in (0,1)$ which is regarded as a stochastic process in the space $\big( D[0,1] \big)^d$ of $\mathbb{R}^d-$valued right-continuous functions with left-hand limits on $[0,1]$ (see, \cite{goh2009nonstandard}, \cite{kato2009asymptotics} and \cite{portnoy2012nearly}\footnote{To establish the limiting behaviour of such stochastic processes several studies in the literature employ the Bahadur representation which holds uniformly in $\theta$ and $\uptau \in (0,1)$. Moreover, estimating the parameter $\theta$ induces some important changes compared to a known $\theta_0$ and required finding an expansion for $\sqrt{n} \left( \widehat{\boldsymbol{\beta}} ( \uptau ; \widehat{\theta} ) - \boldsymbol{\beta} ( \uptau ; \widehat{\theta} ) \right)$.}). 

\medskip

\begin{definition}
(Weak Uniform Convergence on $\Theta$)[see, \cite[p.~26]{white1996estimation}]
Let $( \Omega, \mathcal{F}, \mathcal{P} )$ be a probability space. Then, $Q_n ( \cdot, \theta ) - \bar{Q}_n (  \theta ) \to 0$ as $n \to \infty$ in probability law uniformly on $\Theta$ if for any $\epsilon > 0$ and any $\delta > 0$ there exists an integer $N ( \epsilon, \delta ) < \infty$ such that for all $n > N ( \epsilon, \delta )$ then
\begin{align}
\mathbb{P} \left[ \underset{ \Theta }{ \mathsf{sup} } \left|     Q_n ( \cdot, \theta ) - \bar{Q}_n (  \theta ) \right| < \epsilon \right] > 1 - \delta. 
\end{align} 
\end{definition}

\begin{remark}
If $Q_n$ and $\bar{Q}_n$ are vector or matrix valued functions then we use identical terminology to mean that uniform convergence holds element by element. Moreover, the assumption of compactness of $\Theta$ is crucial since it ensures that the supremum of $\Theta$ exists and that the functions $\underset{ \Theta }{ \mathsf{sup} } \left| Q_n (\cdot, \theta ) - \bar{Q}_n (  \theta ) \right|$ are measurable (see,  \cite{saikkonen1995problems}).

\newpage

In the remaining of this section, we borrow some relevant results from the framework proposed by \cite{bhattacharya2020quantile} (see Appendix of the paper). Within our framework the auxiliary results briefly discussed above can be useful to derive error bounds to establish the order of convergence of remainder terms when implementing a Bahadur representation to quantile-based functionals. Specifically, \cite{bhattacharya2020quantile} considers the  asymptotic behaviour of the following functional 
\begin{align}
\widehat{ \boldsymbol{\mathcal{E} } } \left( \uptau ; \theta \right) 
= 
\sqrt{n} \left( \widehat{\boldsymbol{\beta} } \left( \uptau ; \theta \right) - \boldsymbol{\beta} \left( \uptau ; \theta \right) \right)  - \bigg( \boldsymbol{H}^{-1} \left( \uptau ; \theta \right) \widehat{\boldsymbol{S}} \left( \uptau ; \theta \right)  \bigg). 
\end{align}
A uniform order for $\widehat{\boldsymbol{\mathcal{E} } } \left( \uptau ; \theta \right)$ relies on a uniform order study for the remainder term which crucially depends on the properties of the quantile regression of interest. To do this, one can employ concepts of maximal inequality under bracketing conditions as given in the related literature as well as other linearization techniques (e.g., see \cite{goh2009nonstandard}, \cite{kato2009asymptotics} and \cite{brian2021covering}). 
\end{remark}
In particular, it holds that for any $\uptau \in (0,1)$
\begin{align}
\sqrt{n} \left( \widehat{\beta}(\uptau) - \beta(\uptau)    \right) \overset{ d }{ \to } \mathcal{N} \bigg( 0, V(\uptau) \bigg).
\end{align}
Following the derivations in Proposition 1 of  \cite{bhattacharya2020quantile},  it yields that 
\begin{align*}
\sqrt{n} \left( \widehat{\beta} \big( \uptau ; \widehat{\theta} \big) - \beta(\uptau; \theta_0 ) \right)
&=
\sqrt{n} \bigg\{ \widehat{\beta}(\uptau ; \widehat{\theta} ) - \beta(\uptau; \widehat{\theta} ) \bigg\}
+ 
\sqrt{n} \bigg\{ \beta(\uptau ; \widehat{\theta} ) - \beta(\uptau; \theta_0 ) \bigg\}
\\
&= 
\sqrt{n} \bigg\{ \widehat{\beta}(\uptau ; \widehat{\theta} ) - \beta(\uptau; \widehat{\theta} ) \bigg\}
+ 
\left( \frac{ \partial \beta ( \uptau ; \theta_0 ) }{ \partial \theta } + \mathcal{O}_p(1)  \right) \sqrt{n} \left( \widehat{\theta} - \theta_0 \right)
\\
&=
\sqrt{n} \bigg\{ \widehat{\beta}(\uptau ; \widehat{\theta} ) - \beta(\uptau; \widehat{\theta} ) \bigg\} 
\left( \frac{ \partial \beta ( \uptau ; \theta_0 ) }{ \partial \theta } \right)^{\prime} \frac{1}{\sqrt{n}} \sum_{i=1}^n \Psi (Z_i) + \mathcal{O}_p(1) 
\end{align*}
Therefore, it can be proved that 
\begin{align}
\sqrt{n} \left( \widehat{\beta} \big( \uptau ; \widehat{\theta} \big) - \beta(\uptau; \theta_0 ) \right)
= 
H^{-1} \left( \uptau; \widehat{\theta} \right) \widehat{S} \left( \uptau; \widehat{\theta} \right) 
+ 
\left( \frac{ \partial \beta ( \uptau ; \theta_0 ) }{ \partial \theta } \right)^{\prime} \frac{1}{\sqrt{n}} \sum_{i=1}^n \Psi (Z_i) + \mathcal{O}_p(1). 
\end{align}

\begin{lemma}[\cite{bhattacharya2020quantile}]
Suppose that regularity conditions hold. Then, the following inequalities and unique representations can be derived. 

\begin{itemize}

\item[\textit{\textbf{(i)}}] $Q^{(2)} ( \beta ; \uptau, \theta )$ is continuous with respect to its three arguments, with 
\begin{align}
\norm{ Q^{(2)} \big( \beta_1 ; \uptau, \theta \big) - Q^{(2)} \big( \beta_0 ; \uptau, \theta \big) } \leq C \norm{ \beta_1 -  \beta_2 }
\end{align}
for all $\beta_0$ and $\beta_1, \theta \in \Theta$ and $\uptau \in [0,1]$. 

\item[\textit{\textbf{(ii)}}] $Q^{(2)} ( \beta ; \uptau, \theta )$ is strictly positive for all $\beta \in B(\theta), \theta \in B(\theta)$ and $\uptau \in [0,1]$.

\newpage

\item[\textit{\textbf{(iii)}}] For $\theta \in \Theta$ and $\uptau \in [0,1]$ has a unique minimizer $\beta ( \uptau ; \theta )$ which is continuously diffable in $\theta$ 
\begin{align}
\frac{ \partial \beta ( \uptau ; \theta ) }{ \partial \theta^{\prime} } 
&= 
H ( \uptau ; \theta )^{-1} D ( \uptau ; \theta ),
\\
\frac{ \partial \beta ( \uptau ; \theta ) }{ \partial \theta^{\prime} } 
&= 
H ( \uptau ; \theta )^{-1} \mathbb{E} \left[ X (\theta)  \right]
\end{align}
\end{itemize}
where $H ( \uptau ; \theta )$ and $D ( \uptau ; \theta )$ are defined above.  
\end{lemma}

Consider a uniform order study of the Bahadur error term such that 
\begin{align}
\widehat{ \boldsymbol{\mathcal{E} } } \left( \uptau ; \theta \right) 
&= 
\sqrt{n} \left( \widehat{\boldsymbol{\beta} } \left( \uptau ; \theta \right) - \boldsymbol{\beta} \left( \uptau ; \theta \right) \right)  - \bigg( \boldsymbol{H}^{-1} \left( \uptau ; \theta \right) \widehat{\boldsymbol{S}} \left( \uptau ; \theta \right)  \bigg) 
\\
\underset{ (\uptau, \theta ) \in [ \uptau_1, \uptau_2  ] \times \Theta }{ \mathsf{sup} }  \norm{ \widehat{\boldsymbol{\mathcal{E} } } \left( \uptau ; \theta \right)  }  &= \mathcal{O}_{p} \left( \frac{ \mathsf{log}^{3 / 4} n }{ n^{1/4} } \right).
\end{align}
Define the following optimization function
\begin{align}
\mathcal{L}_n \left( \gamma, \uptau ; \theta \right) = \sum_{t=1}^n \left\{ \rho_{\uptau} \left( Y_t (\theta) - X_t(\theta)^{\top} \left(  \frac{\gamma}{\sqrt{n} } + \boldsymbol{\beta}\left( \uptau ; \theta \right) \right)     \right) - \rho_{\uptau} \bigg( Y_t (\theta) - X_t( \theta)^{\top} \boldsymbol{\beta} \left( \uptau; \theta \right) \bigg) \right\},
\end{align}
This implies that a unique minimizer can be established given by the expression below  
\begin{align}
\sqrt{n} \left( \widehat{ \boldsymbol{\beta} } \left( \uptau; \theta \right) - \boldsymbol{\beta} \left( \uptau; \theta \right) \right) = \underset{ \gamma }{ \mathsf{arg \ min } }  \ \mathcal{L}_n \left( \gamma, \uptau ; \theta \right).
\end{align}
Relevant studies with the topological convergence as well as existence and uniqueness proofs are mentioned in the Appendix of \cite{katsouris2022asymptotic} (see, also  \cite{kato2009asymptotics}). Thus, following the argumentation in \cite{bhattacharya2020quantile} a uniform order for $\widehat{\boldsymbol{\mathcal{E} } } \left( \uptau ; \theta \right)$ relies on a uniform order study for the remainder term $\mathbb{R} \big( \gamma, \epsilon, \uptau ; \theta   \big)$. To do this, one can employ concepts of maximal inequality under bracketing conditions as given in the related literature as well as other linearisation techniques. There are two techniques commonly used in the literature; the bracketing and set covering approach as presented by \cite{bhattacharya2020quantile} and the partitioning approach presented in several studies (see, \cite{brian2021covering} and \cite{goh2009nonstandard}).  The remainder term can be expressed as 
\begin{align}
\mathbb{R}_n \left( \gamma, \epsilon, \uptau ; \theta \right) = 
\mathbb{L}_n \left( \gamma, \epsilon, \uptau ; \theta \right) -
\mathbb{L}_n^0 \left( \gamma, \epsilon, \uptau ; \theta \right) 
= 
\sum_{t=1}^n  \mathbb{R}_t \left( \gamma, \epsilon, \uptau ; \theta \right)
\end{align}

\begin{lemma}[\cite{bhattacharya2020quantile}]
For real numbers $t_{\gamma}$ and $t_{\epsilon} > 0$ with $t_{\epsilon} \geq 1$, $t_{\epsilon} = \big( t \ \mathsf{log}^{3 / 4} n \big) / n^{1 / 4}$ for some $t > 0$, such that $\left( t_{\gamma} + t_{\epsilon} \right)^{1 / 2} / t_{\epsilon} \leq \mathcal{O} \left( n^{1 / 4} / \mathsf{log}^{1 / 2} n  \right)$, for large $n$,  
\begin{align}
\mathbb{E} \left[  \underset{ \left( \gamma, \epsilon, \uptau ; \theta \right) \in \mathcal{B} \left( 0, t_{\gamma} \right)    \times \mathcal{B} \left( 0, t_{\epsilon} \right) \times [ \uptau_1, \uptau_2 ] \times \Theta }{ \mathsf{sup} } \ \bigg| \mathbb{R}_n^1 \big( \gamma, \epsilon, \uptau ; \theta \big)   \bigg| \right] 
\leq C \frac{ \mathsf{log}^{1 / 2} n }{ n^{1/4} } t_{\epsilon} ( t_{\gamma} + t_{\epsilon} )^{1 / 2}.
\end{align}
\end{lemma}

\newpage

\begin{example}
In this example, we discuss the econometric framework of  \cite{chen2021quantile}, who consider quantile regression with generated regressors. Denote with $F_t \left( u | \boldsymbol{x}_t \right)$, the conditional distribution function of the error term  with continuous densities $f_t \left( u | \boldsymbol{x}_t \right)$ and a unique conditional $\uptau-$th quantile which are uniformly bounded away from 0 and $\infty$. It holds that  
\begin{align}
\sqrt{n} \left( \hat{ \boldsymbol{\theta} } - \boldsymbol{\theta}    \right) = n^{- 1/ 2} \sum_{t=1}^n \boldsymbol{r}_t ( \boldsymbol{\theta}) + o_p(1)   
\end{align} 
where $\boldsymbol{r}_t(\cdot)$ is a continuous function which satisfies $\mathbb{E} \left[ \boldsymbol{r}_t ( \boldsymbol{\theta}) \right] = 0$ and Var$\left[ \boldsymbol{r}_t ( \boldsymbol{\theta}) \right] = \boldsymbol{V}$ and thus the generated regressor has the form $
\frac{1}{n} \sum_{t=1}^n \hat{ \boldsymbol{x}}_t \hat{ \boldsymbol{x}}_t^{\prime} - \frac{1}{n}  \sum_{t=1}^n \boldsymbol{x}_t \boldsymbol{x}_t^{\prime} = o_p(1)$. 

Moreover, for any fixed $\uptau \in (0,1)$, we have the following results
\begin{align}
\sqrt{n} \left( \hat{ \boldsymbol{\beta} }(\uptau ) - \boldsymbol{\beta} (\uptau )   \right) \overset{ d }{ \to } \mathcal{N} \big( \boldsymbol{0} , \boldsymbol{\Omega} (\uptau )   \big).
\end{align}
Under the null hypothesis, it holds that
\begin{align}
\sqrt{n} \left( \boldsymbol{R} \hat{ \boldsymbol{\beta} }(\uptau ) - \boldsymbol{r} \right) \overset{ d }{ \to } \mathcal{N} \big( \boldsymbol{0} , \boldsymbol{R} \boldsymbol{\Omega} (\uptau )  \boldsymbol{R}^{\prime} \big).
\end{align}
Then, since we proved that $\hat{\boldsymbol{\Omega}} (\uptau )$ is a consistent estimator of $\boldsymbol{\Omega} (\uptau )$, by Slutsky's theorem,
\begin{align}
\mathcal{W}_T = T \left( \boldsymbol{R} \widehat{ \boldsymbol{\beta} }(\uptau ) - \boldsymbol{r} \right)^{\prime} \big[ \boldsymbol{R} \widehat{ \boldsymbol{\Omega} } (\uptau )  \boldsymbol{R}^{\prime} \big]^{-1} \left( \boldsymbol{R} \widehat{ \boldsymbol{\beta} }(\uptau ) - \boldsymbol{r} \right) \sim \chi^2_q.
\end{align}
Using the mean value expansion theorem we have that 
\begin{align}
\frac{1}{n} \sum_{t=1}^n \hat{ \boldsymbol{x} }_t u_t = \frac{1}{n} \sum_{t=1}^n \boldsymbol{x}_t u_t + \left\{ \frac{1}{n} \sum_{t=1}^n\nabla_{\delta} h \left( \boldsymbol{w}_t, \boldsymbol{\delta} \right)^{\prime} u_t \right\} \big( \hat{ \boldsymbol{\delta} } - \boldsymbol{\delta} \big) + o_p(1). 
\end{align}
Since we have that $\mathbb{E} \left( u | \boldsymbol{w} \right) = 0$ and $\mathbb{E} \left[  \nabla_{\delta} h \left( \boldsymbol{w}_t, \boldsymbol{\delta} \right)^{\prime} u_t \right] = o_p(1)$. Furthermore, since $\big( \hat{ \boldsymbol{\delta} } - \boldsymbol{\delta} \big) = o_p(1)$ and $\mathbb{E} \left(   \boldsymbol{x}_t u_t \right) = 0$, it follows that 
\begin{align}
\frac{1}{n} \sum_{t=1}^n \hat{ \boldsymbol{x} }_t u_t = \frac{1}{n} \sum_{t=1}^n \boldsymbol{x}_t u_t + o_p(1) = o_p(1).
\end{align}
Therefore, it can be proved that $\big( \hat{ \boldsymbol{\beta} } - \boldsymbol{\beta}_0 \big) = o_p(1)$ and the asymptotic normality of the OLS-GR estimator can be also established.  
In particular, we have that 
\begin{align*}
\sqrt{n} \big( \hat{ \boldsymbol{\beta} } - \boldsymbol{\beta}_0 \big)  
&= 
\left( \frac{1}{n} \sum_{t=1}^n  \hat{ \boldsymbol{x} }_t  \hat{ \boldsymbol{x} }_t^{\prime} \right)^{-1} \left\{ \frac{1}{\sqrt{n} } \sum_{t=1}^n \hat{ \boldsymbol{x} }_t \left[ \left( \boldsymbol{x}_t - \hat{\boldsymbol{x}}_t \right)^{\prime}\boldsymbol{\beta}_0 + u_t \right] \right\},
\\
&=  
\widehat{ \boldsymbol{H} }^{-1} \left\{ \frac{1}{\sqrt{n} } \sum_{t=1}^n \hat{ \boldsymbol{x} }_t \left[ \left( \boldsymbol{x}_t - \hat{\boldsymbol{x}}_t \right)^{\prime}\boldsymbol{\beta}_0 + u_t \right] \right\}
\end{align*}
where $\widehat{ \boldsymbol{H} } = \left( \frac{1}{n} \sum_{t=1}^n  \hat{ \boldsymbol{x} }_t  \hat{ \boldsymbol{x} }_t^{\prime} \right)$. 
\end{example}

\newpage

\bibliographystyle{apalike}
\bibliography{myreferences1}

\newpage

\end{document}